\documentclass[sigconf, nonacm]{acmart}



\usepackage[linesnumbered,ruled,vlined]{algorithm2e}

\newcommand{\sstitle}[1]{\vspace*{0.4em}\noindent{\bf #1:\/}}
\usepackage{graphicx}
\usepackage{float}
\usepackage{tikz}
\usetikzlibrary{arrows, arrows.meta, positioning, patterns}

\usepackage{booktabs}
\usepackage{multirow}
\usepackage{tabularray}

\usepackage{amsmath}
\usepackage{amssymb}

\usepackage{amsthm}
\usepackage{thmtools}

\usepackage{xcolor}
\usepackage{pifont}
\usepackage{listings}
\usepackage[normalem]{ulem}
\usepackage{enumitem}
\usepackage{subcaption}
\lstdefinelanguage{SQL}{
  morekeywords={SELECT, FROM, WHERE, JOIN, ON, AND, OR, AS, INNER, LEFT, RIGHT},
  sensitive=false,
  morecomment=[l]{--},
  morestring=[b]',
}
\newcommand{\high}[1]{\uline{#1}}

\newcommand{\unfinished}[1]{\textcolor{red}{#1}}

\usepackage[table]{xcolor}
\definecolor{lightgray}{gray}{0.85}

\newcommand{\higher}[1]{\cellcolor{gray!25}{#1}}

\newcommand{\wcoj}[1]{\textbf{#1}}

\tikzset{
  every node/.append style={
    execute at begin node=$, execute at end node=$
  },
  vtx/.style={
    circle,draw=black,line width=0.8pt,
    minimum size=15pt,inner sep=1.1pt,font=\small
  },
  e/.style={-Stealth,line width=0.8pt,>=Stealth}
}

\newcommand\sysname{\textsf{SplitJoin}}
\newcommand\vldbdoi{XX.XX/XXX.XX}

\newcommand\vldbvolume{14}
\newcommand\vldbissue{1}

\newcommand\vldbavailabilityurl{https://github.com/hyj2003/SplitJoin}
\newcommand\vldbpagestyle{plain} 

\begin{document}
\title{One Join Order Does Not Fit All: Reducing Intermediate Results with Per-Split Query Plans}

\author{Yujun He$^{1,2*}$, Hangdong Zhao$^{3}$, Simon Frisk$^{1}$, Yifei Yang$^{1}$, Kevin Kristensen$^{1}$ \\  Paraschos Koutris$^{1}$, Xiangyao Yu$^{1}$}
\affiliation{
  \institution{$^1$University of Wisconsin-Madison \ \ \ \
  $^2$Southern University of Science and Technology \ \ \ \ 
  $^3$Microsoft Gray Systems Lab}
}
\thanks{$^*$Work done while interning at the University of Wisconsin–Madison.}

\begin{abstract}
Minimizing intermediate results is critical for efficient multi-join query processing. Although the seminal Yannakakis algorithm offers strong guarantees for acyclic queries, cyclic queries remain an open challenge. In this paper, we propose \sysname{}, a framework that introduces \textit{split} as a first-class query operator. By partitioning input tables into heavy and light parts, \sysname{} allows different data partitions to use distinct query plans, with the goal of reducing intermediate sizes using existing binary join engines. We systematically explore the design space for split-based optimizations, including threshold selection, split strategies, and join ordering after splits. Implemented as a front-end to DuckDB and Umbra, \sysname{} achieves substantial improvements: on DuckDB, \sysname{} completes 43 social network queries (vs. 29 natively), achieving 2.1$\times$ faster runtime and 7.9$\times$ smaller intermediates on average (up to 13.6$\times$ and 74$\times$, respectively); on Umbra, it completes 45 queries (vs. 35), achieving 1.3$\times$ speedups and 1.2$\times$ smaller intermediates on average (up to 6.1$\times$ and 2.1$\times$, respectively).
\end{abstract}

\maketitle

\pagestyle{\vldbpagestyle}
\begingroup
\renewcommand\thefootnote{}\footnote{\noindent
This work is licensed under the Creative Commons BY-NC-ND 4.0 International License. Visit \url{https://creativecommons.org/licenses/by-nc-nd/4.0/} to view a copy of this license. For any use beyond those covered by this license, obtain permission by emailing \href{mailto:info@vldb.org}{info@vldb.org}. Copyright is held by the owner/author(s). Publication rights licensed to the VLDB Endowment. \\
\raggedright Proceedings of the VLDB Endowment, Vol. \vldbvolume, No. \vldbissue\ %
ISSN 2150-8097. \\
\href{https://doi.org/\vldbdoi}{doi:\vldbdoi} \\
}\addtocounter{footnote}{-1}\endgroup

\ifdefempty{\vldbavailabilityurl}{}{
\vspace{.3cm}
\begingroup\small\noindent\raggedright\textbf{PVLDB Artifact Availability:}\\
The source code, data, and/or other artifacts have been made available at \url{\vldbavailabilityurl}.
\endgroup
}

\section{Introduction} \label{sec:intro}

Minimizing intermediate join results has been a crucial goal for achieving good performance for multi-join queries. For \textit{acyclic} queries, the seminal Yannakakis algorithm~\cite{DBLP:conf/vldb/Yannakakis81} can provably guarantee intermediate table sizes of at most $O(N+OUT)$, where \textit{N} and \textit{OUT} are the sizes of query input and output, respectively. This can lead to fast and robust query plans in practical systems~\cite{yang2024predicatetransfer, zhao2025robust, yang2025predicate, wang2025yannakakis}. For \textit{cyclic} queries, however, the intermediate table sizes can be much larger than the Yannakakis bound, and thus are challenging to minimize in both theory and practical systems.

Prior work has extensively studied how to address the inefficiency of binary join plans on \textit{cyclic} queries and skewed data (i.e., some values appear many times in a column). Worst-case optimal join (WCOJ) algorithms~\cite{Veldhuizen:ICDT:2014,Ngo2018WCOJ,Salihoglu:SIGMODRec:2023} guarantee runtimes bounded by the maximum possible output size, thereby avoiding the quadratic intermediate blow-ups that can arise in traditional plans. For instance, in the triangle query, a binary join plan may generate $\Theta(N^2)$ intermediate tuples, while a WCOJ algorithm completes in $O(N^{1.5})$ time. Despite this theoretical elegance, WCOJ has not been widely adopted in practice: binary joins remain dominant after decades of engineering optimization, WCOJ implementations typically rely on complex indexing structures, and efficient parallelization is challenging even with recent advances such as HoneyComb~\cite{DBLP:conf/sigmod/WuWZ22}.

Complementary to WCOJs, another line of theoretical work investigates \textit{partition-based strategies}~\cite{AboKhamisNgoSuciu2025,Deep2020FastJoin,Deep2024Output,Hu2024Fast}. These algorithms mitigate skew directly by separating heavy and light values in the join key domains and evaluating them under different subplans. By isolating values with high occurrence, partitioning provides finer control over join orders and reduces intermediate result sizes. In the triangle query, for example, partitioning can also guarantee a $O(N^{1.5})$ bound on intermediate results. Although these strategies are theoretically powerful, their adoption in real systems remains limited due to the complexity of existing approaches~\cite{AboKhamisNgoSuciu2025}, leaving a gap between theory and practice.

Building on the above insights, we take the initial step toward making partition-based strategies practical. Our goal is to bridge the gap between theoretical advances and real-world query engines. To achieve this, we propose \sysname{}, which introduces \textbf{split} as a first-class operator in query processing. The split operator partitions an input table $R$ into a light table ($R_L$) and a heavy table ($R_H$), which are then processed with potentially different query plans.
\sysname{} can substantially reduce intermediates without requiring new data structures like in existing WCOJ algorithms. In this paper, we systematically explore how splits can be incorporated into traditional execution frameworks during query planning. 

Our first contribution is a general framework (Section~\ref{sec:sys_overview}) that incorporates {split} as a first-class citizen in a query plan. This framework reveals a rich design space with several key dimensions, including which table(s) to split, on which attributes to split, how to determine the appropriate split thresholds, and how to decide the best join order on each part of the partition.

From a theoretical perspective, we show in Section~\ref{sec:theory} that our framework can be used to create a large class of provably worst-case optimal algorithms for queries on binary relations that all have the same cardinality. However, achieving this requires strict rules on how to split and which join orders to use. To improve practical performance, we relax these rules in favor of heuristics. 

From a practical perspective, we implement an instantiation of our framework as \sysname{} (Section~\ref{sec:opt}), a front-end layer that can work for any SQL database: it takes a join query as input, generates an optimized query that incorporates split operators where appropriate, and submits it to the underlying database engine for execution. This design allows \sysname{} to take advantage of the existing query processing infrastructure that may support only binary joins while transparently introducing split-based optimizations. We conduct experiments by placing this layer in front of two state-of-the-art databases, DuckDB~\cite{DBLP:conf/cidr/RaasveldtM20} and Umbra~\cite{Freitag2020WCOJ}, and evaluate on commonly used social network datasets, which often exhibit data skew and cyclic query patterns. Our experiments demonstrate that \sysname{} significantly reduces intermediate result sizes and improves query execution time, while preserving the efficiency of carefully optimized binary joins. In particular, DuckDB with \sysname{} successfully completes 43 queries compared to 29 with native DuckDB, and Umbra with \sysname{} completes 45 queries compared to 35 with native Umbra.
For the subset of queries that can finish both with and without \sysname{}, \sysname{} delivers substantial improvements: on DuckDB, queries are 2.1$\times$ faster and produce 7.9$\times$ smaller intermediates on average (up to 13.6$\times$ and 74$\times$, respectively); on Umbra, it achieves 1.3$\times$ speedups and 1.2$\times$ smaller intermediates on average (up to 6.1$\times$ and 2.1$\times$, respectively). Notably, for all query plans generated by \sysname{} in our experiments, we have manually checked and verified that they are provably worst-case optimal, demonstrating that our proposed heuristics work well for the tested queries and data.

The remainder of this paper is organized as follows. Section~\ref{sec:background} introduces the background and motivation of our work, highlighting the limitations of existing solutions. Section~\ref{sec:sys_overview} presents an overview of \sysname{}'s design space. Section ~\ref{sec:theory} provides the theoretical motivation underlying our design. Section ~\ref{sec:opt} describes our current heuristics and optimizations in each design dimension. Section ~\ref{sec:exp} reports experimental findings and Section ~\ref{sec:conclusion} concludes the paper.

\section{Background and Related Work} \label{sec:background}


We focus on the evaluation of natural join queries. As an example, given three relations $R(A,B)$, $S(B,C)$, and $T(C,A)$, our objective could be to evaluate the triangle query, expressed as follows:
\[
Q(A,B,C) \;=\; R(A,B) \bowtie S(B,C) \bowtie T(C,A).
\]
In SQL, this query can be written as:
\begin{lstlisting}[language=SQL]
SELECT R.A, R.B, S.C
FROM R, S, T
WHERE R.B = S.B AND R.A = T.A AND S.C = T.C;
\end{lstlisting}

In this work, we primarily focus on joins over binary relations, i.e., relations with exactly two attributes. 
Joins over binary relations have been extensively studied in prior work~\cite{tao:LIPIcs.ICDT.2020.25,Birler2024Diamond,Ngo2018WCOJ} and provide a foundation that can be extended to more general multi-attribute cases. Moreover, this setting is particularly relevant for graph workloads, where tuples typically represent edges and the data often exhibit skew, which makes binary relations a representative and practically important case to study.

For a (natural) join query $Q$ with binary relations, we define the \textit{query graph} as follows: for every uniquely named attribute $X$ in the query, we create a vertex $X$. For a relation $R(A,B)$, we create an edge between the vertices $A$ and $B$. If $e$ is an edge of the query graph, we use $R_e$ to refer to the relation that corresponds to it.


\subsection{The Problem with Data Skew}

\begin{figure*}[ht]
  \centering
  \subfloat[Binary Join (No data skew)]{
    \includegraphics[height=2.15in]{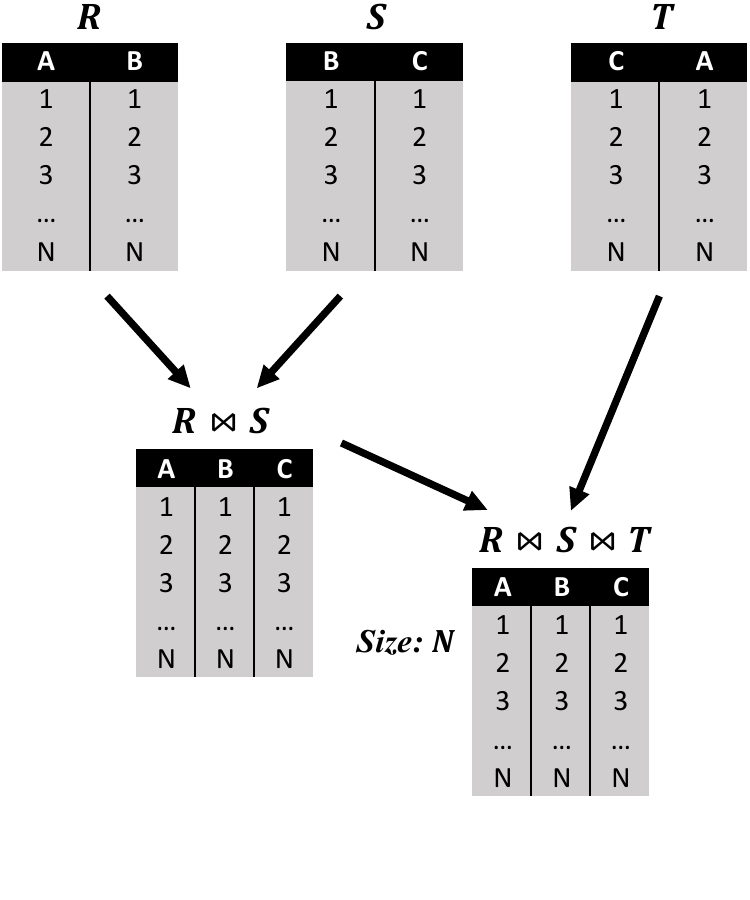}
  }
  \subfloat[Binary Join (Data Skew)]{
    \includegraphics[height=2.15in]{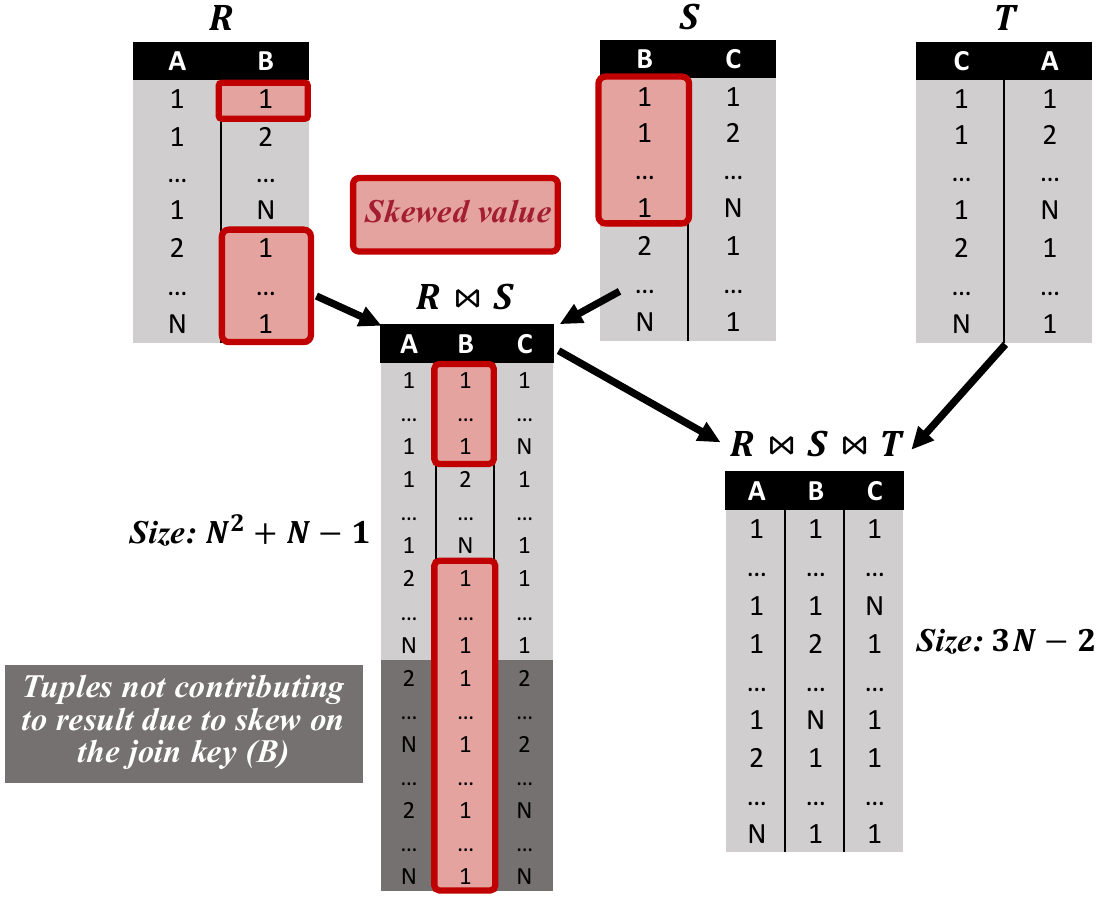}                
  }
  \subfloat[Split-Based Binary Join (Data Skew)]{
    \includegraphics[height=2.15in]{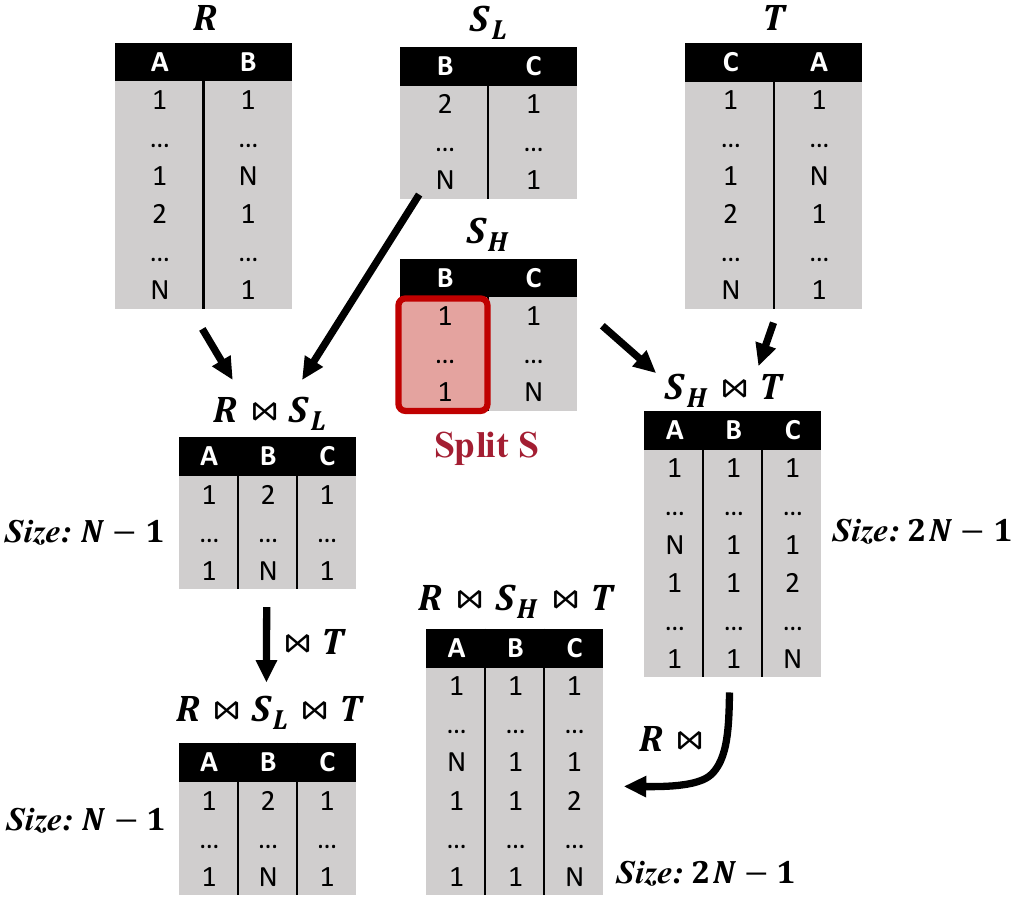}
  }
  \caption{Motivating Example on How Split Avoids Producing Intermediate Tuples That Do Not Contribute to the Final Result.
  }
\label{fig:split-example}
\end{figure*}

Most database engines evaluate join queries using a {\em binary join plan}: this plan forms a tree of join operators, each operator computing the join between two tables. Binary join plans are theoretically optimal when the structure of the query is {\em acyclic}, as shown by Yannakakis~\cite{DBLP:conf/vldb/Yannakakis81}. In these cases, we can find a join plan that runs in time $O(N + OUT)$, where $N$ is the database input size and $OUT$ is the size of the output.

However, for non-acyclic queries (as the triangle query in our example), binary join plans are provably suboptimal~\cite{AGM}. We explain this observation next with an example using the triangle query.

Let's first consider the simple no-skew case where all input relations are as in Figure~\ref{fig:split-example}(a):
\[
R = S = T = \{(1,1), (2,2), \ldots, (N,N)\}.
\] 
We have three possible join orders, and for any order the intermediate relation (which is $R \bowtie S$, $R \bowtie T$, or $S \bowtie T$) has exactly $N$ tuples. This is a good case for binary join, where the intermediate table does not explode in size. 

In contrast, consider the following skewed example (Figure~\ref{fig:split-example}(b)):
\[
R = S = T = \{(1,1), (1,2), (1,3), \ldots, (1,N), (2,1), (3,1), \ldots, (N,1)\}.
\]
In this instance, the data is highly skewed, where in table $R$, the value $A=1$ appears $N$ times and the value $B=1$ also appears $N$ times (similarly for tables $S$ and $T$). This skew leads to a disproportionately large number of intermediate tuples for binary join. Specifically, regardless of the chosen join order, we will create an intermediate table with $\Theta(N^2)$ tuples, while the output size of the entire query is actually $3N-2$, which is linear in $N$. 


To further illustrate why a join result can explode in size, we define the {\em degree} of a value: for a value $a$ in the domain of $A$, its {\em degree} is defined as
$d_R(a) = |\sigma_{A=a} R|$ (i.e., the number of tuples in relation $R$ whose attribute $A$ equals $a$), and define the {\em maximum degree} of attribute $A$ in $R$ as
$\delta_R(A) = \max_{A=a} d_R(a)$ (i.e., the largest degree of any value in a given attribute). 

When joining $R$ with $S$ on attribute $B$, 
a single tuple in $S$ may match up to $\delta_R(B)$ tuples in $R$, 
implying that the output size of the join contributed by that tuple is at most $\delta_R(B)$. Consequently, the join output size of $R$ and $S$ is bounded by $\delta_R(B) \cdot N$.

In the first, no-skew example above, $\delta_R(B) = 1$ and $\delta_S(B) = 1$, 
so the intermediate join result grows linearly with $N$, and no explosion occurs. 
In contrast, in the skewed example, the maximum degree of any attribute in any relation is always $N$. In this case, a single join key value in a relation can match $N$ tuples in the other relation, producing $\Theta(N^2)$ intermediate tuples in the worst case, even though the final output size of the query may be much smaller.

This analysis above illustrates that a fundamental reason of intermediate result explosion in binary join plans is {\em data skew}: high-degree values in the join attributes can lead to disproportionately large intermediate relations, significantly impacting the efficiency of binary join processing.



\subsection{Worst-Case Optimal Joins}

To solve the problem we illustrated above in skewed cyclic queries, a line of work has focused on developing \textit{worst-case optimal join (WCOJ)} algorithms~\cite{Freitag2020WCOJ,Aref2015LogicBlox,Veldhuizen:ICDT:2014,Wang:ProcACM:2023,Salihoglu:SIGMODRec:2023,Feng:CIDR:2023,Arroyuelo:SIGMOD:2021,Brisaboa:SPIRE:2015,Chu:SIGMOD:2015,DBLP:conf/sigmod/WuWZ22,Aberger:SIGMOD:2017,Ngo2018WCOJ}. These algorithms guarantee that the runtime is bounded by the worst-case size of the query output, rather than the potentially much larger intermediate results generated by traditional binary join plans. For example, in theory, all triangle queries can be executed in time $O(N^{1.5})$, instead of $O(N^2)$.
Well-known algorithms such as Leapfrog Triejoin~\cite{Veldhuizen:ICDT:2014} have been successfully adopted in several systems~\cite{Feng:CIDR:2023,Aref2015LogicBlox,Freitag2020WCOJ,Salihoglu:SIGMODRec:2023,Aberger:SIGMOD:2017}, 
and they provide strong theoretical guarantees, where both query runtime and intermediate table sizes can be bounded by $O(N^{1.5})$ for triangle queries.

We identify three key limitations of existing WCOJ solutions. First, these algorithms require sophisticated indexing structures (e.g., hash trie) which are not supported by most database systems today. Second, in real-world workloads, WCOJ algorithms may not always outperform carefully optimized binary joins. Third, these special indexing structures are inherently limited in exploiting modern multi-core architectures and in taking advantage of the well-established efficiency of binary join algorithms~\cite{Ngo2013Skew}. Recent efforts, such as HoneyComb~\cite{DBLP:conf/sigmod/WuWZ22}, have attempted to parallelize WCOJ algorithms, yet the complexity of indexing and limited scalability remain challenging.

\subsection{Partitioning to the Rescue}

Another line of theoretical research explores \textit{partition-based} strategies ~\cite{AboKhamisNgoSuciu2025,Deep2024Output,Deep2020FastJoin,Hu2024Fast} to mitigate the data skew problem. These algorithms separate the ``heavy'' (high-degree) and ``light'' (low-degree) values in the join key domains and evaluate each partition under different subplans. By treating values with high degree independently, these approaches limit the maximum degree we mentioned before and gain finer-grained control over join orders, thus reducing the size of intermediate results. 




Taking the general triangle query for example, assume there are no duplicate tuples in $R$, $S$, and $T$. 
We split $S$ into two tables: $S_H$, consisting of tuples whose $B$-values appear more than $\sqrt{N}$ times, 
and $S_L$, containing the remaining tuples. For $S_H$, there can be at most $\sqrt{N}$ distinct $B$-values. Since no duplicate $(b,c)$ pairs exist, each $C$-value can appear at most $\sqrt{N}$ times---once per distinct $B$-value---with at most $\sqrt{N}$ different $C$-values. 
Consequently, $\delta_{R_H}(C)\leq \sqrt N$, so executing $(S_H(B,C)\bowtie T(C,A))\bowtie R(A,B)$ produces at most $O(N^{1.5})$ intermediate results. For $S_L$, each $B$-value occurs at most $\sqrt{N}$ times, meaning that $\delta_{S_L}(B)\leq \sqrt{N}$, and the total size of a different join order $(R(A,B)\bowtie S_L(B,C))\bowtie T(C,A)$ is also bounded by $O(N^{1.5})$.  

The general triangle query illustrates how the heavy-light partitioning strategy effectively reduces the maximum degree 
of join keys in each sub-table, thereby mitigating intermediate result explosion caused by skew. Returning to our running example, in Figure~\ref{fig:split-example}(c), if we split table $S$ into $S_H=\{(1,1),(1,2), (1,3), \ldots, (1,N)\}$ and $S_L=\{(2,1), (3,1), \ldots, (N,1)\}$ depending on the degree on attribute $A$, and execute $(S_H(B,C)\bowtie T(C,A))\bowtie R(A,B)$ and $(R(A,B)\bowtie S_L(B,C))\bowtie T(C,A)$ separately, the intermediate result size would be only $(2N-1)$ for the former subquery and $(N-1)$ for the latter.

Although partitioning has been well studied in theory, it has not been adopted in real systems due to the complexity of existing approaches. For example, The theoretical state-of-the-art PANDA algorithm~\cite{AboKhamisNgoSuciu2025} further partitions each relation into not only two, but multiple pieces; and each one is part of a different sub-plan. This gap between theoretical benefits and practical applicability motivates us to ask the following question, which drives this work: {\em how can we leverage \textbf{split} as a new operator in query processing to solve the data skew problem? }

\subsection{Other Related Work}

\sstitle{Free Join}
Free Join~\cite{Remy2023FreeJoin} made a step toward reconciling the theoretical guarantees of WCOJ with the practical advantages of binary hash-join execution. It provides a unified framework that, in principle, can capture both robustness and efficiency. In practice, however, Free Join suffers from several shortcomings: the current implementation is restricted to a single-threaded setting, its reliance on specialized index structures complicates deployment, and the current optimization procedure does not always produce high-quality plans. These challenges limit its applicability in modern query engines, leaving the question open of how to systematically close the gap between theory and practice. 

\sstitle{Lookup \& Expand}
A recent work aims to minimize intermediate results in join processing by decomposing a traditional hash join into two sub-operations: \emph{Lookup} and \emph{Expand}~\cite{Birler2024Diamond}. In this framework, Lookup verifies the existence of a matching tuple in the join partner relation without immediately materializing all join partners, thereby filtering out dangling tuples early. Expand, in contrast, is invoked only when it becomes necessary to enumerate all matches, which postpones the creation of large intermediate results. This separation enables more robust join execution and opens opportunities for reordering these two operators in the query plan. The Lookup \& Expand paradigm may be further enhanced when combined with \sysname{}, which further limits the intermediate result size.


\section{The Splitting Framework} \label{sec:sys_overview}



\tikzset{
  box/.style={draw, thick, minimum width=1.5cm, minimum height=0.5cm, align=center},
  redbox/.style={draw, thick, fill=red!20, minimum width=2.5cm, minimum height=1cm, align=center},
  arrow/.style={->, thick}
}

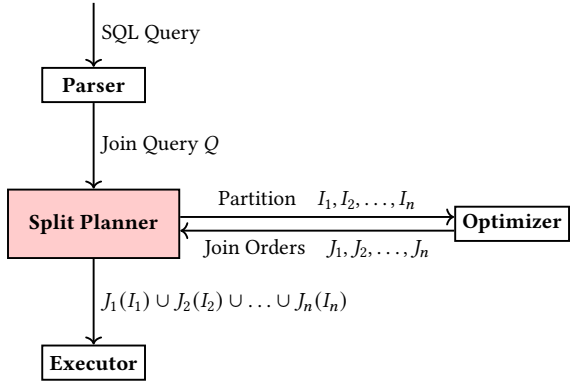
\begin{figure}[t]
\centering
\resizebox{0.9\linewidth}{!}{
    \begin{tikzpicture}[node distance=1.25cm]
    
    \node[box] (parser) {{\bf Parser}};
    \node[redbox, below=of parser] (split) {{\bf Split Planner}};
    \node[box, below=of split] (executor) {{\bf Executor}};
    \node[box, right=4cm of split] (optimizer) {{\bf Optimizer}};
    
    \draw[arrow] (parser) -- (split) node[midway, right] {\text{Join Query }Q};
    \draw[arrow] (split) -- (executor) node[midway, right] {J_1(I_1) \cup J_2(I_2) \cup \ldots \cup J_n(I_n)};
    \draw[arrow] ([yshift=0.1cm] split.east) -- ([yshift=0.1cm] optimizer.west) node[midway, above] {\text{Partition}  \quad I_1, I_2, \ldots, I_n};
    \draw[arrow] ([yshift=-0.1cm] optimizer.west) -- ([yshift=-0.1cm] split.east) node[midway, below] {\text{Join Orders} \quad J_1, J_2, \ldots, J_n};
    
    \draw[arrow] (0,1.2) -- (parser.north) node[midway, right] {\text{SQL Query}};
    \end{tikzpicture}
}
\caption{Overview of the Split Planner.}
\label{fig:overview}
\end{figure}


In this section, we present a general framework that incorporates splitting into query evaluation.

The framework extends the query optimizer with an additional component called a \emph{split planner} (see Figure~\ref{fig:overview}). The split planner consists of two phases: the {\em split phase} (Algorithm ~\ref{algo:split}) and the {\em join phase} (Algorithm~\ref{algo:join}). This component is integrated into the optimizer so that query planning is not only aware of join orders and access paths, but also of possible split strategies.

The split phase takes as input the join query $Q$, a database instance $I$ (i.e., the set of tables referenced by $Q$), and a set of splits $\Sigma$.  A {\em split} consists of a pair $(\mathcal{R},A)$, where $\mathcal{R}$ is a set of relations, and $A$ is an attribute that appears in all relations $R \in \mathcal{R}$. For every split $(\mathcal{R},A)$, we split the values of the attribute $A$ into two subsets: light ($A_L$) and heavy ($A_H$). Then, every relation $R \in \mathcal{R}$ is split into a light ($R_L$) and heavy ($R_H$) part according to the split of the join attribute. This creates two subinstances $I_L, I_H$. Each of the two subinstances is then recursively partitioned using the remaining splits in the split set $\Sigma$ until $\Sigma$ becomes empty. In general, we want to make sure that every relation is split at most one time in $\Sigma$ (some relations may not be split at all). The end result is a partition of $I$ in smaller subinstances $I_1, I_2, \dots$

Algorithm ~\ref{algo:split} is parametrized by two functions. The first function is \textit{splitAttribute}, which determines the heavy and light values of the common (join) attribute. The second is a function \textit{chooseSplitSet}, which will select a split set $\Sigma$ that will be used as the initial input. We will discuss their possible implementation in the next two sections.

The join phase takes as input the subinstances $I_1, I_2, \dots$ generated from the split phase, and then decides an appropriate join order for each subinstance. Finally, it computes the union of the results to form the final query output. One important consideration in the join phase is that the optimizer needs to incorporate the properties of this subinstance (e.g., what is light, what is heavy) in order to find the best plan. We will also discuss how this can be achieved in the next two sections.



\begin{algorithm}[t]
\caption{\textbf{Split Phase}}
\label{algo:split}
\KwIn{Join query $Q$, instance $I$, split set $\Sigma$}
\KwOut{a partition $\{I_1, I_2, \dots\}$}

\If{$\Sigma = \emptyset$}{
    \Return $\{I$\}}
$(\mathcal{R},A) \gets$ a split in $\Sigma$ \;
$(A_L, A_H) \gets \textit{splitAttribute}(\mathcal{R},A)$ \;
\For{every $R \in \mathcal{R}$}{
$R_L \gets \{t \in R \mid t.A \in A_L \}$ \;
$R_H \gets \{t \in R \mid t.A \in A_H \}$ \;
}
$I_L \gets I$ with $\mathcal{R}$ replaced by $\{R_L \mid R \in \mathcal{R}\}$ \;
$I_H \gets I$ with $\mathcal{R}$ replaced by $\{R_H \mid R \in \mathcal{R}\}$ \;
$\Sigma \gets \Sigma \setminus \{(\mathcal{R},A)\}$ \;
\Return $\textbf{Split}(Q, I_L,\Sigma) \cup\textbf{Split}(Q, I_H,\Sigma)$
\end{algorithm}

\begin{algorithm}[t]
\caption{\textbf{Join Phase}}
\label{algo:join}
\KwIn{Join query $Q$, partition $\mathcal{I} = \{I_1, I_2, \dots \}$ }
\KwOut{Q(I)}

    

\For{each part $I_i \in \mathcal{I}$}{
  compute $Q(I_i)$ using a skew-aware query optimizer
}
\Return $\bigcup_i Q(I_i)$
\end{algorithm}


\section{Theoretical Motivation}  
\label{sec:theory}

We next present an instantiation of our splitting framework for binary relations that is provably worst-case optimal. In particular, it achieves the AGM bound $\mathcal{AGM}(Q) = N^\rho$, where $\rho$ is the minimum fractional edge cover of the query graph. We next discuss how to instantiate the split and join phase to achieve the desired bound.


\paragraph{Split Phase} We pick a split set $\Sigma$ as follows. For every relation $R$, we choose either one of its attributes (say $A$), and add $(\{R\},A)$ to the split set. In other words, we split every relation using one of its attributes -- this means we will end up with $2^\ell$ subinstances ($\ell$ is the number of relations). Then, we choose $A_L = \{a \mid d_R(a) \leq \sqrt{N} \}$ and $A_H = \{a \mid d_R(a) > \sqrt{N} \}$. In other words, we split the values according to their degree in $R$ and a threshold $\tau = \sqrt{N}$. This threshold choice means that if we split a relation $R(A,B)$ into $R_L, R_H$ using attribute $A$, the degree of any value of $A$  will be light in $R_L$ (smaller than $\sqrt{N}$), while the degree of any value of $B$ will be light in $R_H$. We can depict this visually by thinking the query graph as a directed graph, where an edge is directed away from the light attribute. Since we split every relation, each subinstance corresponds to one way of making the query graph directed.

\paragraph{Join Phase} We now discuss how we pick a good join order for a subinstance $I_i$. Recall from above that we can think of this subinstance as a directed query graph $G_i$ of $Q$. The key intuition we want to follow is: {\em try to join on light attributes as much as possible}. 
We define this formally via an operation called a {\em light join}. Given an intermediate result $T(X_1, \dots, X_k)$ in a query plan, a light join $T(X_1, \dots, X_k)\bowtie R(X_k,Y)$ is a join where the edge $(X_k, Y)$ is in $G_i$, i.e. $X_k$ is light in relation $R$.

In the below pseudocode, we describe the join order. Intuitively, we do the following: pick any relation to start with (Line 2), and perform as many light joins as possible (Lines 5-7). If no light joins are possible, start over somewhere else and repeat. When there are two intermediate results that reach each other, they are merged (Lines 8-11). When we are left with one intermediate result, the query output is obtained (Line 14).

The algorithm below works on queries that are connected, i.e. do not contain cartesian products. Cartesian products are handled by first computing each connected sub-query.
\begin{algorithm}
\caption{Join Ordering}
\label{algorithm:BinaryJoin}
\KwData{$Q$: Query, $I_i$: Input}
$\mathcal{C} \gets \{\}$ ; \quad// no intermediate results yet \\
\While{$\exists$ attribute $X$ not in any set in $\mathcal{C}$ s.t. an unused relation $R(X,Y)$ is light in $X$}{
$c \gets \{X, Y\}$ \;
$T_c \gets R(X,Y)$ \;
\While{$\exists$ light join between $T_c$ and an unused relation $R(X,Y)$ with $X \in c$}{
$T_{c \cup \{Y\}} \gets T_c \bowtie R(X,Y)$ \;
$c \gets c \cup \{Y\}$ \;
}
\While{$\exists c' \in \mathcal{C}$ such that $c \cap c' \neq \emptyset$}{
$T_{c \cup c'} \gets T_c \bowtie T_{c'}$ \;
$c \gets c \cup c'$ \;
$\mathcal{C} \gets \mathcal{C} - \{c'\}$ \;
}
$\mathcal{C} \gets \mathcal{C} \cup c$
}
// $\mathcal{C} = \{c\}$, i.e. it has a single component \\
\Return $T_c$ \;
\end{algorithm}

\begin{example}\label{example:TheoryAlgExample}
    Consider the following query
    \begin{align*}
    Q =& 
    R_1(A,B) \bowtie
    R_2(B,C) \bowtie
    R_3(A,C) \bowtie
    R_4(C,D) \bowtie \\
    &R_5(A,D) \bowtie
    R_6(D,E) \bowtie
    R_7(E,K) \bowtie 
    R_8(B,F) \bowtie \\
    &R_9(F,G) \bowtie
    R_{10}(G,H) \bowtie
    R_{11}(H,I) \bowtie
    R_{12}(H,J) \bowtie \\
    &R_{13}(I,K) \bowtie
    R_{14}(J,K) \bowtie
    R_{15}(C,I) \bowtie
    R_{16}(B,H)
    \end{align*}

    The query graph is depicted in the figure below. The color/number indicates which intermediate relation a relation is first joined into by line 5-7.
    \begin{figure}[t]
    \begin{center}\begin{tikzpicture}[every node/.style={circle, draw}]
        \draw[red,rounded corners=10pt,thick] (-0.7,-2.2) rectangle (2.2,2.2);
        \node[draw=none] at (-0.4,1.9) {\textbf{1}};
        \draw[blue,rounded corners=10pt,thick] (2.3,-0.7) rectangle (5.2,2.2);
        \node[draw=none] at (2.6,1.9) {\textbf{2}};
        \draw[green,rounded corners=10pt,thick] (0.8,2.3) rectangle (3.7,3.6);
        \node[draw=none] at (1.1,3.3) {\textbf{3}};
        \draw[orange,rounded corners=10pt,thick] (2.3,-0.8) rectangle (3.7,-2.2);
        \node[draw=none] at (2.6,-1.1) {\textbf{4}};
        \node[draw=none] at (1.1,3.3) {\textbf{3}};
        \node[fill=white] (a) at (0,0) {A};
        \node[fill=white] (b) at (1.5,1.5) {B};
        \node[fill=white] (c) at (1.5,0) {C};
        \node[fill=white] (d) at (1.5,-1.5) {D};
        \node[fill=white] (e) at (3,-1.5) {E};
        \node[fill=white] (f) at (1.5,3) {F};
        \node[fill=white] (g) at (3,3) {G};
        \node[fill=white] (h) at (3,1.5) {H};
        \node[fill=white] (i) at (3,0) {I};
        \node[fill=white] (j) at (4.5,1.5) {J};
        \node[fill=white] (k) at (4.5,0) {K};
        \draw[thick,->,red] (a) -> (b);
        \draw[thick,->,red] (b) -> (c);
        \draw[thick,->,red] (c) -> (a);
        \draw[thick,->,red] (c) -> (d);
        \draw[thick,->,red] (d) -> (a);
        \draw[thick,->,green] (f) -> (b);
        \draw[thick,->,orange] (e) -> (d);
        \draw[thick,->,orange] (e) -> (k);
        \draw[thick,->,green] (f) -> (g);
        \draw[thick,->,green] (g) -> (h);
        \draw[thick,->,blue] (k) -> (i);
        \draw[thick,->,blue] (i) -> (h);
        \draw[thick,->,blue] (h) -> (j);
        \draw[thick,->,blue] (j) -> (k);
        \draw[thick,->,blue] (h) -> (b);
        \draw[thick,->,blue] (i) -> (c);
    \end{tikzpicture}\end{center}
    \caption{Example query graph and initial intermediate relations}
    \end{figure}
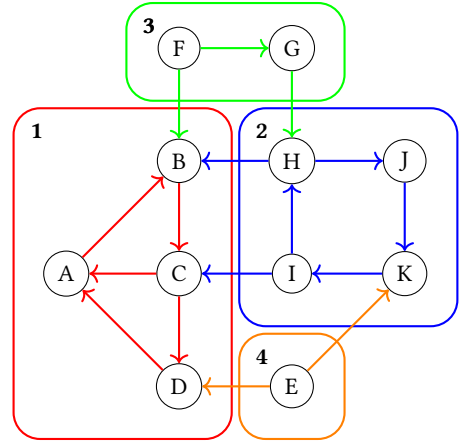
    Suppose the algorithm start doing light joins from $A$. This will yield an intermediate relation with the red (1) relations, since this intermediate has no outgoing edges (light joins).
    Suppose that next, the algorithm starts expanding from $H$. It will expand an intermediate relation of the blue (2) relations. The red and blue intermediates will then be merged.
    Continuing, the green (3) intermediate is computed and merged with the red/blue intermediate. Finally, the same happens with the orange (4) component.
\end{example}

We prove the following theorem in Appendix \ref{sec:TheoryAppendix}.

\begin{restatable}[]{theorem}{TheoryAlgMatchAGM}
Let $Q$ be a natural join query with binary relations of at most size $N$. Then, the above instantiation of the splitting framework runs in time at most $\mathcal{AGM}(Q) = O(N^{\rho})$.
\end{restatable}

We give some intuition behind the theorem. For any intermediate that is created only with light joins, it is never bigger than $\mathcal{AGM}(Q)$, because light joins grow the intermediate very slowly (the intermediate grows by at most a factor $\sqrt{N}$ for every light join). For queries over binary relations, this is enough to match $\mathcal{AGM}(Q)$. The trickier part is understanding why the merge step matches the AGM bound. It is because the intermediate results have a form that is close to \textit{induced subgraph form}, which is discussed more in the appendix.

\paragraph{Discussion}
This section has proven that there is a large class of query plans with splits that match the AGM bound. Indeed, several join orders are produced from Algorithm~\ref{algorithm:BinaryJoin}, since we are free to choose what is our starting point and which light join we do first.
In practice, these theoretically optimal plans may not be desirable --- they split every table, which is not necessary and expensive in practice. In the next section, we will see how we can split more carefully. Additionally, we would like to give the query optimizer more freedom in picking among join orders during the join phase. However, we want to keep the same intuition: light joins are less likely to grow an intermediate result too much and should be given preference. One alternative interpretation of our theoretical result is that by splitting, the query optimizer is much more likely to pick a good join order for each subinstance, since the join orders that satisfy the AGM bound cover a much larger space.

\section{SplitJoin Design Space} \label{sec:opt}



While the theoretical algorithm described in Section~\ref{sec:theory} ensures worst-case optimality, it may not lead to the best performance in practice. In this section, we explore the design space of \sysname{} and propose heuristics that can lead to good practical performance. 

\subsection{Co-split}
The first design decision of \sysname{} is to pick splits of the form $(\{R,T\},A)$, i.e., we always split two relations at once. To explain why this is beneficial compared to splitting one relation at a time, consider a query $Q$ that has two tables $R(A,B), T(A,C)$ that join on attribute $A$.  If we choose to split $R$ on attribute $A$, then we create two subrelations $R_L$ and $R_H$. Now observe that since $R,T$ join on attribute $A$, this means that in the subinstances $I_L, I_H$ it suffices to keep only the parts of $T$ that join with $R_L, R_H$ respectively:
\[
\begin{aligned}
T_L(A,C) &\gets T \Join \pi_A(R_L) \\
T_H(A,C) &\gets T \Join \pi_A(R_H)
\end{aligned}
\]

In other words, a split on $R$ can be thought as a split on $T$ as well; this can reduce the input size for each join by ensuring that related tuples are split consistently. The co-split strategy makes this explicit. One additional consideration is that we want to avoid splitting both $R$ and $T$ on attribute $A$ independently, since this would generate four subproblems: $R_L \bowtie T_L$, $R_L \bowtie T_H$, $R_H \bowtie T_L$, and $R_H \bowtie T_H$.  

We can visualize co-splits alternatively by considering the {\em join graph} of a query $Q$. This graph can be thought as the "dual" of the query graph: each vertex is a relation, and two vertices (relations) are connected via an edge only if they join on some attribute.  Figure~\ref{fig:case_study} shows as example the join graph of query $Q_5$ (defined in Section~\ref{sec:setup} and Figure~\ref{fig:queries}). A co-split then corresponds to an edge of the join graph. For instance, the co-split $(\{R_1, R_5\},Y)$ corresponds to the edge $\{R_1,R_5\}$ of the join graph in our running example. It will be convenient to use the shorthand $R_1 \Join_Y R_5$ for a co-split.

\begin{figure}[t]
  \centering
  \includegraphics[width=\linewidth]{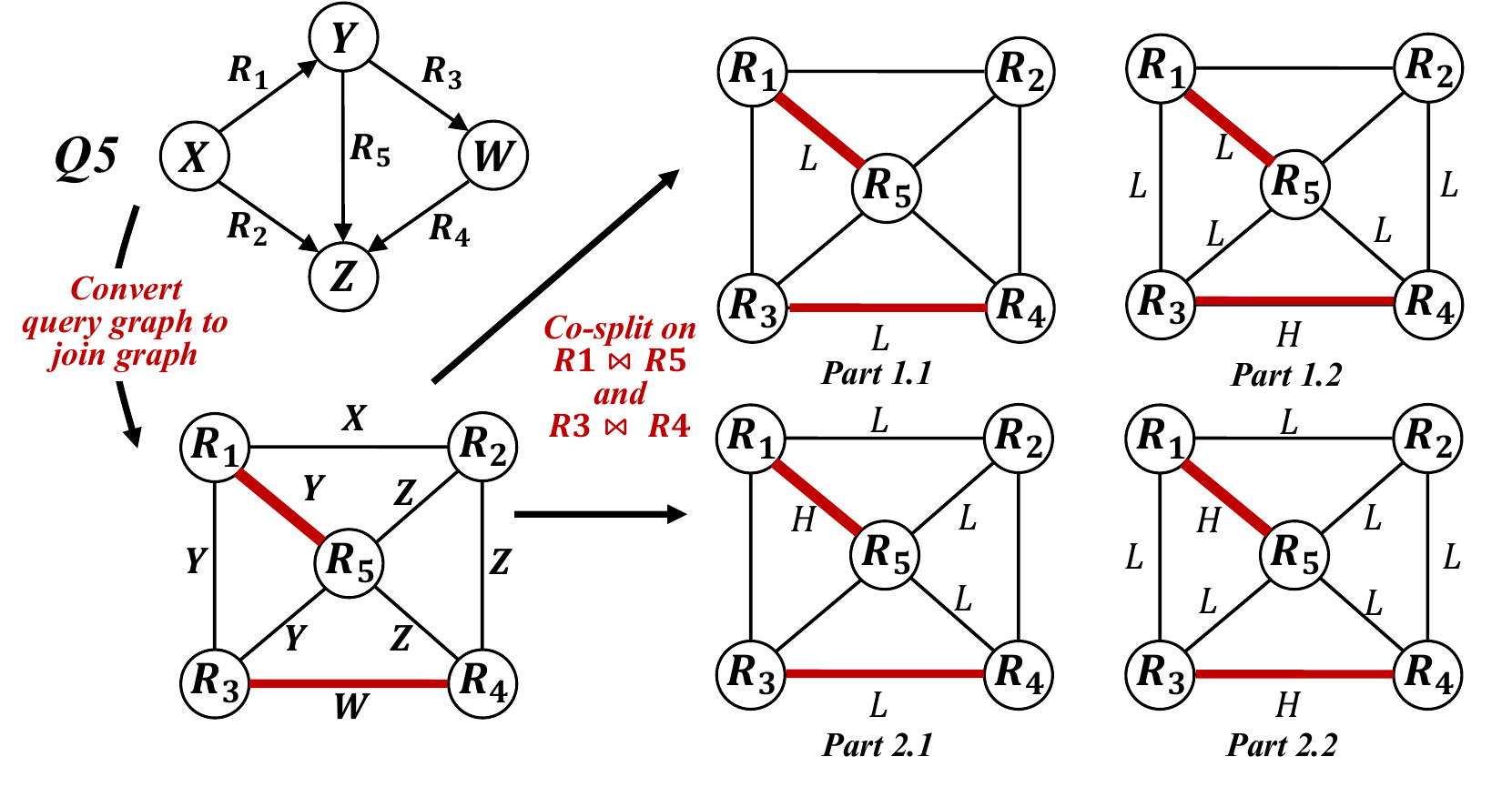}
  \captionsetup{font=small,skip=3pt}
  \caption{Split Example on Q5, L/H indicates a join is Light/Heavy.}
  \label{fig:case_study}
\end{figure}

\subsection{Choosing the Split Threshold} \label{sec:split_threshold}


In this part, we discuss how we implement the method \textit{splitAttribute}. We will start by presenting the method for a single split, and then extend this to a co-split. 

\paragraph{Single Split} Suppose we want to split the values of attribute $A$ in a relation $R(A,B)$ of size $N$ into a heavy part $A_H$ and a light part $A_L$. Recall the {\em degree} defined in Section~\ref{sec:background}, we pick a {\em split threshold} value $\tau >0$ such that a value $a$ with $ d_R(a) > \tau$ is sent to the heavy part, and with  $ d_R(a) \leq  \tau$ to the light part. The key question is: {\em how do we choose the right split threshold?}

To answer this question, we need to dive deeper into what is achieved by splitting a relation into two parts. 

\begin{itemize}
    \item \textbf{Light part ($A_L$):} The maximum degree of any value of attribute $A$ is less than $\tau$. This implies that if we join any relation with $R_L$ using $A$ as the join key, the intermediate result will increase by a factor of at most $\tau$.
    \item \textbf{Heavy part ($A_H$):} The main observation is that the number of heavy values of $A$ is at most $N/\tau$. Thus, the maximum degree of attribute $B$ in $R_H(A,B)$ is at most $N/\tau$. This similarly implies that joining with $R_H$ on $B$ will increase the intermediate result by a factor of at most $N / \tau$.
\end{itemize}

Theoretically, we want to balance the two parts, so we would choose $\tau$ such that $\tau = N/\tau$, or equivalently $\tau = \sqrt{N}$ as in Section~\ref{sec:theory}. However, this threshold matches only the worst-case bound. In practice, the bound of $\sqrt{N}$ for the number of heavy values can be overly conservative, as the actual degree distribution on real-world datasets often deviates from the theoretical worst-case assumptions.
Figure~\ref{fig:distribution} illustrates the degree distributions from several real-world network datasets, which are the datasets in Table~\ref{tab:datasets}. As shown in the figure, the vast majority of values have degrees much less than $\sqrt{N}$. Therefore, by replacing the $\sqrt N$ threshold with a smaller value $K$, we can reduce the overhead caused by values with degree less than $\sqrt N$ while preserving the worst-case optimality.

\begin{figure}[t]
  \centering
  \includegraphics[width=\linewidth]{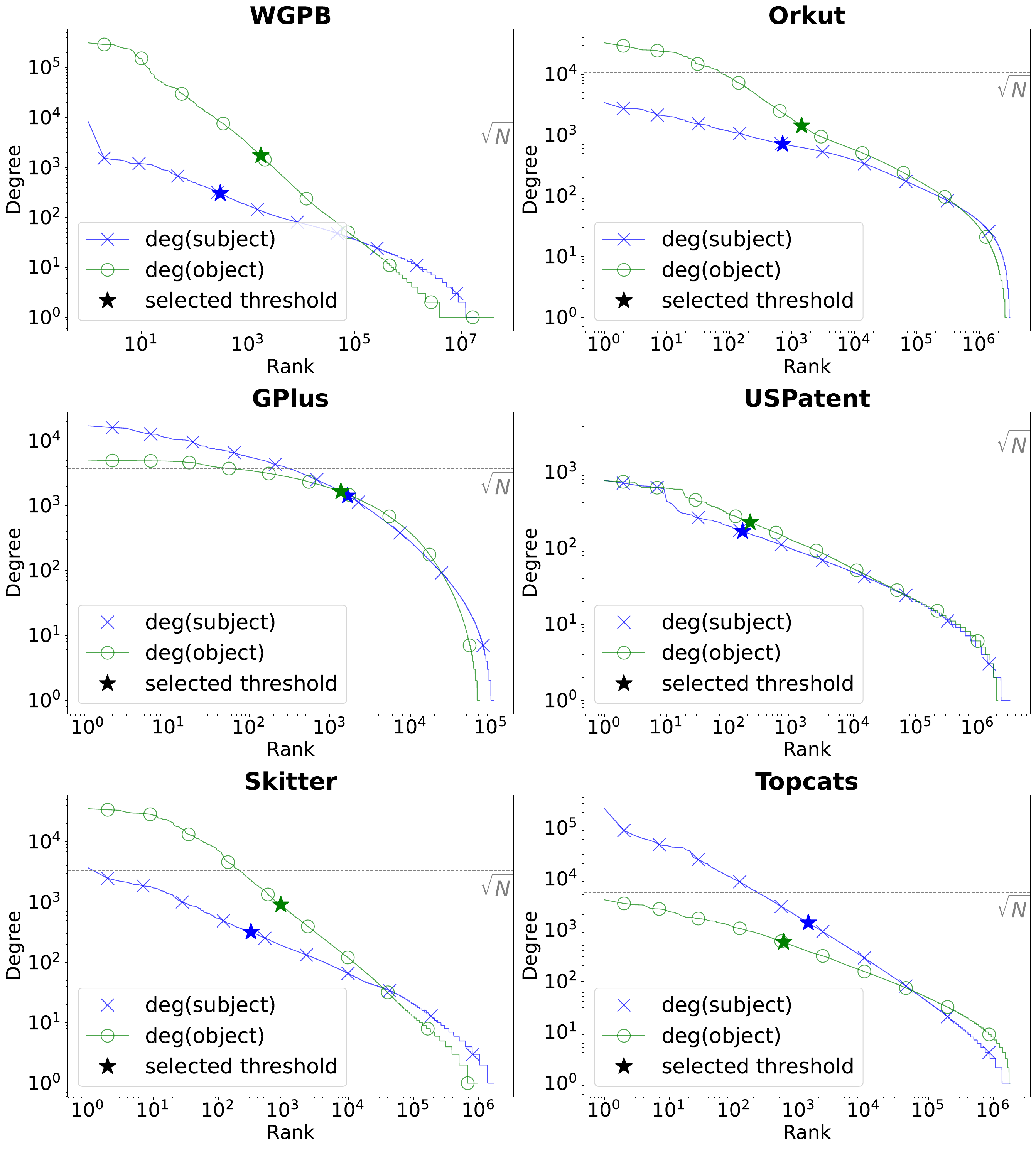}
  \captionsetup{font=small,skip=3pt}
-  \caption{Degree distributions of six datasets. Nodes are ordered by descending degree, with the subject attribute and the object attribute distributions shown on a log–log scale.}
  \label{fig:distribution}
\end{figure}


 We now discuss our method to choose the threshold $\tau$. First, we compute the {\em degree sequence}~\cite{DBLP:journals/pacmmod/KhamisNOS24} of $A$ in $R(A,B)$. The degree sequence is constructed by calculating the degree $d_R(a)$ of every value $a$ in $A$, and sorting them in non-increasing order: 
$$\deg_1 \ge \deg_2 \ge \deg_3 \ge \dots \ge \deg_m$$
Here, $m$ is the number of distinct values in $A$. We want to guarantee that we pick $\tau$ such that the number of values with degree $>\tau$ are at most $\tau$. Hence, we pick $\tau$ to be the first index $K \in \{1, \dots, m\}$ in the degree sequence such that $K \geq \deg_K$. Observe that this choice of threshold is sensitive to the specific data distribution, in contrast to the choice of $\sqrt{N}$ which is data-independent. In particular, if the degree distribution is highly skewed (say a zipfian distribution), then $\deg_K$ will be significantly smaller than $\sqrt{N}$, leading to a more precise and effective splitting. Figure~\ref{fig:distribution} shows the chosen threshold of our heuristic compared to the choice of $\sqrt{N}$. Observe that the chosen threshold is smaller than $\sqrt{N}$, which means that we are much more conservative in what we consider a "light" value.




The split is designed to handle skew by partitioning tuples into heavy and light sides, but this process is not free, since it involves multiple overheads: $(i)$ the cost of partitioning a relation into heavy and light, and  $(ii)$  the cost of executing additional join branches for each partition. When the output size of a join is already small, the cost of these extra steps can easily outweigh any potential performance gains from splitting.
To address this issue, we tweak the threshold choice in one more way. Specifically, we pick two predefined parameters $\Delta_1, \Delta_2$, and check whether:
$$ \deg_1  / \Delta_1 \leq K \leq \Delta_2$$
If this condition is satisfied, then we put everything in the light part (this is equivalent to setting the threshold $\tau = +\infty$).
Intuitively, the second inequality tells us that there are very few heavy values, and the first one that these few heavy values are not too skewed. In this case, partitioning is unnecessary because of the overheads so we opt to skip it. We choose $\Delta_1 = 5$ and $\Delta_2 = 240$ in our experiments.


\paragraph{Co-Split} Suppose now we want to choose the threshold for a co-split $R \Join_A T$. Instead of constructing two degree sequences, we construct a combined degree sequence by considering the combined degree $d_{R,T}(a)=\min \{d_R(a),d_T(a)\}$. We similarly choose the threshold $K$ to be the smallest index such that $K \geq d_{R,T}(a)_K$, and also use $\Delta_1, \Delta_2$ to further optimize the threshold choice exactly as in the case of a single split.  

Moreover, it is also important to understand how heaviness distributes across related joins. As shown in Figure~\ref{fig:case_study}, given that a join is heavy, we can infer that its adjacent joins involving different attributes are light, since the number of distinct attributes in those joins is limited. For example, in part 1.2, if $R_3 \bowtie R_4$ is identified as heavy, then $R_3 \bowtie R_1$, $R_3 \bowtie R_5$, $R_4 \bowtie R_5$, and $R_4 \bowtie R_2$ can be inferred to be light. This inference helps us reason about the broader join structure and identify which parts of the query are likely to contribute less to the overall cost.

\subsection{Choosing the Split Set} 
\label{sec:split_schedule}

We next turn to the problem of deciding the split set. Recall that our split set consists of co-splits, which can be viewed as edges in the join graph. In order to minimize partitioning, we choose the edges such that they form an {\em edge packing} of the graph (i.e., we cannot choose two edges that are adjacent to the same table). This guarantees that each relation will be split at most once. We will say that an edge $\{R,S\}$ in the join graph is {\em uncovered} by a split set $\Sigma$ if $R,S$ do not occur in $\Sigma$.

Specifically, we construct all possible co-split sets through a recursive enumeration procedure, $\textsf{enum}(\Sigma)$. Initially, we start with $\Sigma = \emptyset$, so we call $\textsf{enum}(\emptyset)$. At each recursive step, to compute $\textsf{enum}(\Sigma)$, we first generate all $\Sigma \cup \{R \Join_A S\}$ such that the new co-split ${R \Join_A S}$ satisfies two conditions:
$(i)$ ${R, S}$ is an uncovered edge in the join graph w.r.t. $\Sigma$; and
$(ii)$ there is no other uncovered edge whose two relations belong to a smaller cycle in the query graph.
The function then returns the union of the split sets produced recursively by each $\textsf{enum}(\Sigma \cup \{R \Join_A \})$. The latter condition makes sure that we prioritize edge packings that contain edges in smaller cycles. When $\Sigma$ is an edge packing (i.e., there are no more uncovered edges), we return $\textsf{enum}(\Sigma) = \{\Sigma\}$ and the recursion terminates.

\begin{example}
Consider again $Q_5$ (Figure~\ref{fig:case_study}).
The split set construction prioritizes co-splitting joins within the sub-cycles $R_1(X,Y) \bowtie R_2(X,Z) \bowtie R_5(Z,Y)$ and $R_5(Z,Y) \bowtie R_3(Y,U) \bowtie R_4(U,Z)$, rather than on the joins $R_1 \bowtie R_3$ and $R_2 \bowtie R_4$ that occur in a cycle of length 4. The reason we want to do this is that the former joins each contain a filter-table, allowing the filter effect to be applied immediately during splitting, whereas the latter do not provide this benefit. Hence, we will obtain the following possible co-split sets from the enumeration procedure:
\allowdisplaybreaks
\begin{align*}
   \Sigma_1 = &  \{R_1 \Join_Y R_5, R_3 \Join_W R_4\} \\
   \Sigma_2 = & \{R_2 \Join_Z R_5, R_3 \Join_W R_4\} \\
    \Sigma_3 =& \{R_1 \Join_X R_2, R_3 \Join_W R_4\} \\
   \Sigma_4 = & \{R_1 \Join_X R_2, R_3 \Join_Y R_5\} \\
   \Sigma_5 = & \{R_1 \Join_X R_2, R_4 \Join_Z R_5\}
\end{align*}
Note that the co-splits $R_1 \Join_Y R_3$ and $R_2 \Join_Z R_4$ will never be chosen by our construction.
\end{example}

Given all the generated split sets, the question is now: which set do we choose for splitting? Our strategy is to prioritize splitting on \emph{light joins}, i.e., joins whose threshold is small.

Technically, we will assign a {\em cost} to each edge in the join graph: the cost is the threshold $k$ we choose for this co-split. The key idea here is that the intermediate size of the join is bounded by $k \cdot N$, where $N$ is the size of the largest relation involved. The cost of a set of co-split is simply the maximum of the thresholds of each co-split. Therefore, choosing the joins with the smallest threshold yields the tightest upper bound on intermediate result size, and is expected to produce the most efficient plan.

\begin{example}
Consider again the running example for $Q_5$, and assume that the threshold for the co-split $R_i \Join R_j$ is $k_{ij}$. Then the cost of $\Sigma_1$ would be $\max\{k_{15}, k_{34}\}$, the cost of $\Sigma_2$ would be $\max\{k_{25},k_{34}\}$, and so on. In the case of Figure~\ref{fig:case_study}, the smallest cost is $\Sigma_1$, so we would pick this co-split to partition the instance into the four subinstances than can be seen in the figure.

\end{example}

\subsection{Split-aware Query Optimizer}

Conventional query optimizers in existing database systems are generally unaware of splits. For example, when evaluating triangle queries in DuckDB or Umbra, the optimizer only considers standard join ordering and cost models, without explicitly accounting for the structural changes introduced by splits. Besides, these query optimizers do not utilize degree information. As a result, the execution strategies are often suboptimal, and the potential benefits of splitting cannot be fully exploited. 

Recent work~\cite{LpBound} incorporates degree information into cost estimation, but its approach relies on solving complex linear programs and assumes precise degree distributions, which limits its practicality for general-purpose optimizers. In particular, it incurs high estimation overhead and cannot handle nested queries, which are common in our split-based plans. Instead, we adopt a simpler yet effective heuristic that directly leverages observed maximum degrees and thresholds to guide the cost model.

Building on this idea, we design a simple split-aware optimizer that explicitly incorporates degree information into query planning. In particular, our approach estimates the cost of query 
plans by combining the maximum degree of attributes before splitting with the thresholds determined after splitting. This enables the optimizer to better capture the effect of splits and to generate more efficient execution plans. As an example, if we have partitioned $R(A,B)$ on attribute $A$ with threshold $K$, the cost of the join $T(A,C) \bowtie R_L(A,B)$ can be estimated as $|T| \cdot K$ (since the maximum degree of $A$ in $R$ is at most $K$). Similarly, the cost of the join $S(B,C) \bowtie R_H(A,B)$ is at most $|S| \cdot K$ (since there are at most $K$ heavy values of $A$). 

Apart from this modification, the overall optimization procedure still follows the conventional 
dynamic programming (DP) framework used in query optimizers. That is, plan enumeration and 
cost-based pruning remain unchanged, while only the cost estimation step is extended to be aware of splits using the maximum degree and thresholds. This design ensures compatibility with existing optimizers while enabling them 
to effectively leverage the advantages introduced by splitting.

\section{Experimental Evaluation}\label{sec:exp}

In this section, we present experimental results to demonstrate the effectiveness and advantage of \sysname{}. We begin by briefly describing our implementation in ~\ref{sec:impl}, then state the experimental setup in Section~\ref{sec:setup}, followed by evaluations of overall performance in Section~\ref{sec:overall}, then demonstrate our proposed optimizations by effectiveness studies in Section~\ref{sec:effective}.

\subsection{Implementation}\label{sec:impl}

\sstitle{Preprocessing} To determine the thresholds required for splits, we first collect degree information by using aggregate clause 
for each table in the input query. For each column of every relation, we construct a compact summary table consisting of pairs $(\text{value}, \text{degree})$, where \texttt{degree} denotes the number of occurrences of a given attribute value. To keep the overhead manageable, we only record the at most top 100K values with the highest degrees. This summary table is typically very small, amounting to less than 0.1\% of the size of the original relation, while still being sufficient to capture the skew patterns that drive effective splitting decisions. 

\sstitle{Front-end Layer for Splits} To integrate \sysname{} into existing systems with minimal modification, we designed a lightweight
front-end layer that operates independently of the underlying database engine. This layer takes a join query as input and rewrites it into an optimized SQL query that incorporates split operators before passing it to the database. In addition, this front-end layer can leverage the underlying execution engine to collect the statistical information required by preprocessing, ensuring that degree summaries are obtained efficiently without adding additional data management components.  
The approach is simple to implement and nonintrusive, as it does not require altering the executor internals.

\subsection{Experimental Setup}\label{sec:setup}

\sstitle{Hardware Platform} 
We conduct all experiments on a CloudLab~\cite{cloudlab} machine with r6525 instance type with two AMD EPYC 7543 processors 
(32 cores per socket, 2 hardware threads per core, totaling 128 hardware threads) 
running at 2.8\,GHz, equipped with 251\,GB DDR4 memory. 
The used SSD is a Dell Ent NVMe AGN MU U.2 with 1.6\,TB capacity.

\sstitle{Tested Dataset} 
The datasets used in our experiments are listed in Table~\ref{tab:datasets}. 
They are the same graph datasets evaluated in~\cite{DBLP:conf/sigmod/WuWZ22}, 
covering a range of sizes, sparsity levels, and degree distributions.

\begin{table}[ht]
\centering
\begin{tabular}{lrrl}
\hline
Name & Nodes & Edges & Features\\ 
\hline
WGPB~\cite{WGPB2022,wgpb2} & 54.0M & 81.4M & skew, sparse \\
Orkut~\cite{10.1145/1298306.1298311} & 3.07M & 117M & uniform, partial dense \\
GPlus~\cite{10.5555/2999134.2999195}& 107K & 13.6M & skew, dense\\
USPatent~\cite{10.1145/1081870.1081893} & 3.77M & 16.5M & uniform, sparse\\
Skitter~\cite{10.1145/1081870.1081893} & 1.69M & 11.1M & partial skew, sparse\\
Topcats~\cite{10.1145/3097983.3098069} & 1.79M & 28.5M & skew, partial dense\\
\hline
\end{tabular}
\caption{Sizes of the tested graph datasets.}
\label{tab:datasets}
\end{table}

For the WGPB dataset, which originally contains three columns, we discard the second column as it is not used as a join key in our queries, and remove duplicates from the remaining two columns (about 1\% reduction regarding the number of rows).

\sstitle{Tested Queries}
The queries we used are listed in Figure ~\ref{fig:queries}, which are cyclic queries with a number of tables less than nine from the paper~\cite{10.14778/3342263.3342643}. Additionally, we add the 5-cycle query.

\begin{figure}[htbp]
\centering
\begin{tikzpicture}[x=1.25cm,y=1.25cm]

\begin{scope}[shift={(0,0)}]
  \node (a1) [vtx] at (0,0.5)  {a_1};
  \node (a2) [vtx] at (0.5,0)  {a_2};
  \node (a3) [vtx] at (-0.5,0) {a_3};
  \draw[e] (a1) -- (a2);
  \draw[e] (a2) -- (a3);
  \draw[e] (a3) -- (a1);
  \node[font=\small] at (0,-0.5) {(1) Q1.};
\end{scope}

\begin{scope}[shift={(1.6,0)}]
  \node (a1) [vtx] at (-0.5,0.5)   {a_1};
  \node (a2) [vtx] at (0, 1)   {a_2};
  \node (a3) [vtx] at (0, 0)  {a_3};
  \node (a4) [vtx] at (0.5, 0.5)   {a_4};
  \draw[e] (a1) -- (a2);
  \draw[e] (a1) -- (a3);
  \draw[e] (a3) -- (a4);
  \draw[e] (a2) -- (a4);
  \node[font=\small] at (0,-0.5) {(2) Q2.};
\end{scope}

\begin{scope}[shift={(3.2,0)}]
  \node (a1) [vtx] at (-0.5,0.5)   {a_1};
  \node (a2) [vtx] at (0, 1)   {a_2};
  \node (a3) [vtx] at (0, 0)  {a_3};
  \node (a4) [vtx] at (0.5, 0.5)   {a_4};
  \draw[e] (a1) -- (a2);
  \draw[e] (a1) -- (a3);
  \draw[e] (a4) -- (a3);
  \draw[e] (a2) -- (a4);
  \node[font=\small] at (0,-0.5) {(3) Q3.};
\end{scope}

\begin{scope}[shift={(4.8,0)}]
  \node (a1) [vtx] at (-0.5,0.5)   {a_1};
  \node (a2) [vtx] at (0, 1)   {a_2};
  \node (a3) [vtx] at (0, 0)  {a_3};
  \node (a4) [vtx] at (0.5, 0.5)   {a_4};
  \draw[e] (a1) -- (a2);
  \draw[e] (a1) -- (a3);
  \draw[e] (a3) -- (a4);
  \draw[e] (a2) -- (a4);
  \draw[e] (a2) -- (a3);
  \node[font=\small] at (0,-0.5) {(4) Q4.};
\end{scope}

\begin{scope}[shift={(0,-2)}]
  \node (a1) [vtx] at (-0.5,0.5)   {a_1};
  \node (a2) [vtx] at (0, 1)   {a_2};
  \node (a3) [vtx] at (0, 0)  {a_3};
  \node (a4) [vtx] at (0.5, 0.5)   {a_4};
  \draw[e] (a1) -- (a2);
  \draw[e] (a1) -- (a3);
  \draw[e] (a4) -- (a3);
  \draw[e] (a2) -- (a4);
  \draw[e] (a2) -- (a3);
  \node[font=\small] at (0,-0.5) {(5) Q5.};
\end{scope}

\begin{scope}[shift={(1.6,-2)}]
  \node (a1) [vtx] at (-0.5,0.5)   {a_1};
  \node (a2) [vtx] at (0, 1)   {a_2};
  \node (a3) [vtx] at (0, 0)  {a_3};
  \node (a4) [vtx] at (0.5, 0.5)   {a_4};
  \draw[e] (a1) -- (a2);
  \draw[e] (a1) -- (a3);
  \draw[e] (a3) -- (a4);
  \draw[e] (a2) -- (a4);
  \draw[e] (a2) -- (a3);
  \draw[e] (a1) -- (a4);
  \node[font=\small] at (0,-0.5) {(6) Q6.};
\end{scope}

\begin{scope}[shift={(3.2,-2)}]
  \node (a1) [vtx] at (-0.5,1)   {a_1};
  \node (a2) [vtx] at (-0.5, 0)   {a_2};
  \node (a3) [vtx] at (0, 0.5)  {a_3};
  \node (a4) [vtx] at (0.5, 1)   {a_4};
  \node (a5) [vtx] at (0.5, 0)   {a_5};
  \draw[e] (a1) -- (a2);
  \draw[e] (a1) -- (a3);
  \draw[e] (a3) -- (a4);
  \draw[e] (a4) -- (a5);
  \draw[e] (a2) -- (a3);
  \draw[e] (a3) -- (a5);
  \node[font=\small] at (0,-0.5) {(7) Q7.};
\end{scope}

\begin{scope}[shift={(4.8,-2)}]
  \node (a1) [vtx] at (-0.5,1)   {a_1};
  \node (a4) [vtx] at (-0.5, 0)   {a_4};
  \node (a3) [vtx] at (0, 0.5)  {a_3};
  \node (a2) [vtx] at (0.5, 1)   {a_2};
  \node (a5) [vtx] at (0.5, 0)   {a_5};
  \node (a6) [vtx] at (1, 0.5)   {a_6};
  \draw[e] (a1) -- (a2);
  \draw[e] (a3) -- (a1);
  \draw[e] (a2) -- (a3);
  \draw[e] (a3) -- (a4);
  \draw[e] (a5) -- (a3);
  \draw[e] (a4) -- (a5);
  \draw[e] (a6) -- (a2);
  \draw[e] (a6) -- (a5);
  \node[font=\small] at (0,-0.5) {(8) Q8.};
\end{scope}

\begin{scope}[shift={(0,-4)}]
  \node (a1) [vtx] at (-0.5,0.5)   {a_1};
  \node (a2) [vtx] at (0,1)   {a_2};
  \node (a3) [vtx] at (0,0)   {a_3};
  \node (a4) [vtx] at (0.5,0.5)   {a_4};
  \node (a5) [vtx] at (1,1)   {a_5};
  \node (a6) [vtx] at (1,0)   {a_6};
  \draw[e] (a1) -- (a2);
  \draw[e] (a2) -- (a4);
  \draw[e] (a1) -- (a3);
  \draw[e] (a3) -- (a4);
  \draw[e] (a4) -- (a5);
  \draw[e] (a5) -- (a6);
  \draw[e] (a4) -- (a6);
  \node[font=\small] at (0.4,-0.5) {(9) Q9.};
\end{scope}

\begin{scope}[shift={(2,-4)}]
  \node (a1) [vtx] at (-0.4,0.8)   {a_1};
  \node (a2) [vtx] at (0.4,0.8)   {a_2};
  \node (a3) [vtx] at (1.2,0.8)   {a_3};
  \node (a4) [vtx] at (1.2,0)   {a_4};
  \node (a5) [vtx] at (0.4,0)   {a_5};
  \node (a6) [vtx] at (-0.4,0)   {a_6};
  \draw[e] (a1) -- (a2);
  \draw[e] (a2) -- (a3);
  \draw[e] (a3) -- (a4);
  \draw[e] (a4) -- (a5);
  \draw[e] (a5) -- (a6);
  \draw[e] (a6) -- (a1);
  \node[font=\small] at (0.4,-0.5) {(10) Q10.};
\end{scope}

\begin{scope}[shift={(4,-4)}]
  \node (a1) [vtx] at (-0.15,0.4)   {a_1};
  \node (a2) [vtx] at (0.4,0.8)   {a_2};
  \node (a3) [vtx] at (1.1,0.8)   {a_3};
  \node (a4) [vtx] at (1.1,0)   {a_4};
  \node (a5) [vtx] at (0.4,0)   {a_5};
  \draw[e] (a1) -- (a2);
  \draw[e] (a2) -- (a3);
  \draw[e] (a3) -- (a4);
  \draw[e] (a4) -- (a5);
  \draw[e] (a5) -- (a1);
  \node[font=\small] at (0.4,-0.5) {(11) Q11.};
\end{scope}

\end{tikzpicture}
\caption{Tested Queries}
\label{fig:queries}
\end{figure}
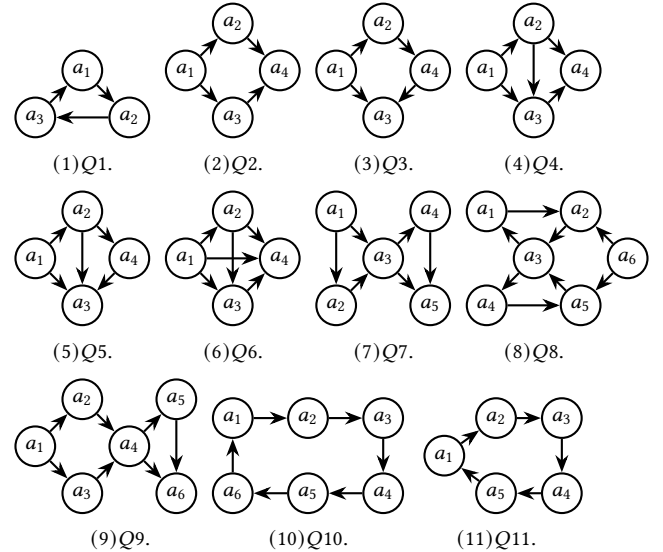
\sstitle{Tested Systems} 
We test our method in two representative database systems: DuckDB~\cite{DBLP:conf/cidr/RaasveldtM20} and Umbra~\cite{Freitag2020WCOJ}. 

DuckDB (v1.3.0) is an open-source in-process columnar DBMS optimized for OLAP workloads, providing easy integration into analytical applications within the  same process. DuckDB primarily relies on optimized binary joins and has no built-in worse-case optimal joins.
Umbra (25.07.1) is a high-performance DBMS designed for modern hardware. It uses a compact columnar format and morsel-driven execution for parallelism. Umbra natively supports worst-case optimal joins, enabling efficient multiway join execution for dense graph queries.

\sstitle{Metrics}
We measure the overall query execution time and intermediate table sizes for all engines.
For \sysname{}, we report end-to-end execution time that includes the time for obtaining thresholds, performing table partitioning, and executing the query. For DuckDB, we use its memory mode (all tables stored in memory, no disk I/O). For Umbra, we use the default configuration, except for adjustments to query memory consumption, timeout, and the WCOJ-related options described below. 

For both \sysname{} and the baselines, each query is executed four times, and we report the minimum execution time to avoid cold start effects. We set a memory limit of 220 GB and a timeout of 15 minutes for each run, and use 32 threads for query execution.

\subsection{Overall Performance Evaluation}\label{sec:overall}To assess the overall performance benefits introduced by the strategies proposed in the previous sections, we conduct a comprehensive comparison between \sysname{} and the baseline join algorithms within each database system.

\subsubsection{DuckDB Comparison}  
Table~\ref{tab:duckdb_cmp} reports the runtime comparison between \sysname{} over DuckDB and the original DuckDB binary join, and Table~\ref{tab:duckdb_max_intermediate_mt} shows the maximum intermediate result sizes. Overall, all queries that can be executed by the original DuckDB (29 in total) can also be executed with \sysname{}. Moreover, with \sysname{}, DuckDB is able to complete 14 additional queries that the original binary join fails to finish due to timeouts or out-of-memory (43 vs. 29). For the queries that can finish in both cases, \sysname{} reduces runtime by 2.1$\times$ on average (up to 13.6$\times$) and cuts the maximum intermediate results by 7.9$\times$ on average (up to 74$\times$).

For datasets with high skew (i.e., WGPB, GPlus, and Topcats), \sysname{} achieves substantial performance gains. On WGPB and Topcats, DuckDB frequently fails due to timeouts or out-of-memory, while \sysname{} completes most queries successfully with significantly lower runtime. For GPlus, both approaches fail on the majority of queries, but for the completed query (Q1), \sysname{} can achieve a speedup of 1.38$\times$  over the original DuckDB. These results highlight the effectiveness of our strategy in coping with highly skewed data distributions, which are prevalent in real-world social network datasets.

For the partially skewed dataset (i.e., Skitter), \sysname{} still demonstrates clear benefits. Because Skitter is relatively sparse and many joins have small thresholds, a large portion of splits can be skipped. Consequently, most of the queries degenerate to DuckDB’s original binary join plan. For the remaining finished queries (Q3, Q4), \sysname{} achieves improvements over the baseline in Q4; however, for Q3, the performance is slightly worse than DuckDB’s original plan. The slowdown primarily stems from the extra cost introduced by the split operator, which lengthens the execution pipeline and causes one of the subqueries to run slower than the baseline query, even though the total size of intermediate results is smaller than in the baseline.

For uniform datasets (i.e., Orkut and USPatent), the performance benefit largely depends on graph density. On USPatent, the gap between \sysname{} and the DuckDB baseline is relatively small, since all joins in the query graph are lightweight, leading \sysname{} to consistently adopt the original binary join plans. In contrast, on Orkut, \sysname{} outperforms the baseline in most queries except Q7. The exception arises because Q7 involves two triangles, $(R_1\bowtie R_2\bowtie R_3)\bowtie(R_4\bowtie R_5\bowtie R_6)$, and the optimal strategy would be to compute each triangle independently and then join the intermediate results. In practice, however, our execution generates four subqueries that cannot fully reuse overlapping computations, thereby causing performance degradation on this query.

Taken together, these results confirm that \sysname{} provides consistent and often substantial runtime improvements, especially under skewed data distributions where traditional binary joins struggle the most.

\begin{table*}[t]
\centering
\small
\begin{tabular}{l l c c c c c c c c c c c}
\toprule
 &  & Q1 & Q2 & Q3 & Q4 & Q5 & Q6 & Q7 & Q8 & Q9 & Q10 & Q11 \\
\midrule
\multirow{2}{*}{WGPB} 
 & \sysname{} & \higher{3.824} & \higher{29.879} & \higher{25.095} & \higher{16.074} & \higher{59.876} & \higher{19.824} & \higher{19.824} & \unfinished{\texttt{OOM}} & \unfinished{\texttt{TLE}} & \unfinished{\texttt{OOM}} & \higher{48.775} \\
 & Default & \unfinished{\texttt{TLE}} & 59.230 & \unfinished{\texttt{OOM}} & \unfinished{\texttt{TLE}} & \unfinished{\texttt{TLE}} & \unfinished{\texttt{OOM}} & \unfinished{\texttt{TLE}} & \unfinished{\texttt{OOM}} & \unfinished{\texttt{TLE}} & \unfinished{\texttt{OOM}} & 666.440 \\
\midrule
\multirow{2}{*}{Orkut} 
 & \sysname{} & \higher{15.076} & \unfinished{\texttt{OOM}} & \unfinished{\texttt{OOM}} & \higher{57.86} & \higher{56.548} & \higher{101.028} & 567.008 & \higher{128.968} & \unfinished{\texttt{TLE}} & \unfinished{\texttt{OOM}} & \unfinished{\texttt{OOM}} \\
 & Default & 19.200 & \unfinished{\texttt{OOM}} & \unfinished{\texttt{OOM}} & \unfinished{\texttt{OOM}} & \unfinished{\texttt{OOM}} & \unfinished{\texttt{TLE}} & \higher{515.240} & \unfinished{\texttt{OOM}} & \unfinished{\texttt{OOM}} & \unfinished{\texttt{OOM}} & \unfinished{\texttt{OOM}} \\
\midrule
\multirow{2}{*}{GPlus} 
 & \sysname{} & \higher{27.156} & \unfinished{\texttt{OOM}} & \unfinished{\texttt{OOM}} & \unfinished{\texttt{TLE}} & \unfinished{\texttt{TLE}} & \unfinished{\texttt{OOM}} & \unfinished{\texttt{OOM}} & \unfinished{\texttt{OOM}} & \unfinished{\texttt{OOM}} & \unfinished{\texttt{OOM}} & \unfinished{\texttt{TLE}} \\
 & Default & 34.770 & \unfinished{\texttt{OOM}} & \unfinished{\texttt{OOM}} & \unfinished{\texttt{OOM}} & \unfinished{\texttt{OOM}} & \unfinished{\texttt{OOM}} & \unfinished{\texttt{OOM}} & \unfinished{\texttt{OOM}} & \unfinished{\texttt{OOM}} & \unfinished{\texttt{OOM}} & \unfinished{\texttt{OOM}} \\
\midrule
\multirow{2}{*}{USPatent} 
 & \sysname{} & \high{0.365} & \high{1.036} & \high{1.167} & \high{1.906} & \high{1.676} & \high{13.907} & \high{1.437} & \high{1.336} & \high{2.536} & \high{1.877} & \high{0.886} \\
 & Default & \high{0.355} & \high{1.010} & \high{1.070} & \high{1.840} & \high{1.850} & \high{14.280} & \high{1.370} & \high{1.200} & \high{2.690} & \high{1.840} & \high{0.910} \\
\midrule
\multirow{2}{*}{Skitter} 
 & \sysname{} & \high{0.741} & \high{7.748} & 4.732 & \higher{2.753} & \high{4.632} & \high{\unfinished{\texttt{TLE}}} & \high{72.324} & \high{7.094} & \high{\unfinished{\texttt{TLE}}} & \high{\unfinished{\texttt{OOM}}} & \high{29.716} \\
 & Default & \high{0.706} & \high{7.750} & \higher{4.220} & 4.820 & \high{4.360} & \high{\unfinished{\texttt{TLE}}} & \high{72.300} & \high{6.950} & \high{\unfinished{\texttt{TLE}}} & \high{\unfinished{\texttt{OOM}}} & \high{29.770} \\
\midrule
\multirow{2}{*}{Topcats} 
 & \sysname{} & \higher{4.525} & \higher{24.056} & \higher{22.917} & \higher{7.976} & \higher{9.207} & \higher{11.836} & \higher{47.456} & \higher{671.069} & \unfinished{\texttt{TLE}} & \unfinished{\texttt{OOM}} & \higher{433.366} \\
 & Default & 7.800 & 30.660 & 248.680 & \unfinished{\texttt{OOM}} & \unfinished{\texttt{OOM}} & \unfinished{\texttt{TLE}} & 384.250 & \unfinished{\texttt{OOM}} & \unfinished{\texttt{OOM}} & \unfinished{\texttt{OOM}} & 803.250 \\
\bottomrule
\end{tabular}
\caption{Runtime (s) in DuckDB. \unfinished{\texttt{TLE}} indicates 900s time limit exceeded. \unfinished{\texttt{OOM}} indicates out-of-memory.  For comparison, the best performance is highlighted with a \colorbox{gray!25}{gray background}, and \high{underline} indicates both systems use the same query plan.}
\label{tab:duckdb_cmp}
\vspace{-8pt}
\end{table*}

\begin{table*}[t]
\centering
\small
\begin{tabular}{l l c c c c c c c c c c c}
\toprule
 & & Q1 & Q2 & Q3 & Q4 & Q5 & Q6 & Q7 & Q8 & Q9 & Q10 & Q11 \\
\midrule
\multirow{2}{*}{WGPB} 
 & \sysname{} & \higher{559} & \higher{559} & \higher{1026} & \higher{1026} & \higher{2982} & \higher{559} & \higher{559} & \unfinished{\texttt{OOM}} & \unfinished{\texttt{TLE}} & \unfinished{\texttt{OOM}} & \higher{7345} \\
 & Default & \unfinished{\texttt{TLE}} & 936 & \unfinished{\texttt{OOM}} & \unfinished{\texttt{TLE}} & \unfinished{\texttt{TLE}} & \unfinished{\texttt{OOM}} & \unfinished{\texttt{TLE}} & \unfinished{\texttt{OOM}} & \unfinished{\texttt{TLE}} & \unfinished{\texttt{OOM}} & 26706 \\
\midrule
\multirow{2}{*}{Orkut} 
 & \sysname{} & \higher{9450} & \unfinished{\texttt{OOM}} & \unfinished{\texttt{OOM}} & \higher{11803} & \higher{9450} & \higher{16452} & \higher{9450} & \higher{9450} & \unfinished{\texttt{TLE}} & \unfinished{\texttt{OOM}} & \unfinished{\texttt{OOM}} \\
 & Default & 10913 & \unfinished{\texttt{OOM}} & \unfinished{\texttt{OOM}} & \unfinished{\texttt{OOM}} & \unfinished{\texttt{OOM}} & \unfinished{\texttt{TLE}} & 57171 & \unfinished{\texttt{OOM}} & \unfinished{\texttt{OOM}} & \unfinished{\texttt{OOM}} & \unfinished{\texttt{OOM}} \\
\midrule
\multirow{2}{*}{GPlus} 
 & \sysname{} & \higher{3286} & \unfinished{\texttt{OOM}} & \unfinished{\texttt{OOM}} & \unfinished{\texttt{TLE}} & \unfinished{\texttt{TLE}} & \unfinished{\texttt{OOM}} & \unfinished{\texttt{OOM}} & \unfinished{\texttt{OOM}} & \unfinished{\texttt{OOM}} & \unfinished{\texttt{OOM}} & \unfinished{\texttt{TLE}} \\
 & Default & 6949 & \unfinished{\texttt{OOM}} & \unfinished{\texttt{OOM}} & \unfinished{\texttt{OOM}} & \unfinished{\texttt{OOM}} & \unfinished{\texttt{OOM}} & \unfinished{\texttt{OOM}} & \unfinished{\texttt{OOM}} & \unfinished{\texttt{OOM}} & \unfinished{\texttt{OOM}} & \unfinished{\texttt{OOM}} \\
\midrule
\multirow{2}{*}{USPatent} 
 & \sysname{} & \high{82} & \high{82} & \high{237} & \high{234} & \high{237} & \high{11528} & \high{179} & \high{80} & \high{234} & \high{194} & \high{194} \\
 & Default & \high{82} & \high{82} & \high{237} & \high{234} & \high{237} & \high{11528} & \high{179} & \high{80} & \high{234} & \high{194} & \high{194} \\
\midrule
\multirow{2}{*}{Skitter} 
 & \sysname{} & \high{454} & \high{454} & 454 & 496 & \high{454} & \high{\unfinished{\texttt{TLE}}} & \high{30693} & \high{454} & \high{\unfinished{\texttt{TLE}}} & \high{\unfinished{\texttt{OOM}}} & \high{18774} \\
 & Default & \high{454} & \high{454} & 454 & \higher{454} & \high{454} & \high{\unfinished{\texttt{TLE}}} & \high{30693} & \high{454} & \high{\unfinished{\texttt{TLE}}} & \high{\unfinished{\texttt{OOM}}} & \high{18774} \\
\midrule
\multirow{2}{*}{Topcats} 
 & \sysname{}   & \higher{1944} & \higher{1944} & \higher{1944} & \higher{1944} & \higher{1944} & \higher{1994} & \higher{1944} & \higher{2371} & \unfinished{\texttt{TLE}} & \unfinished{\texttt{OOM}} & \higher{199772} \\
 & Default & 2663 & 2663 & 144725 & \unfinished{\texttt{OOM}} & \unfinished{\texttt{OOM}} & \unfinished{\texttt{TLE}} & 144725 & \unfinished{\texttt{OOM}} & \unfinished{\texttt{OOM}} & \unfinished{\texttt{OOM}} & 238984\\ 
\bottomrule
\end{tabular}
\caption{Max intermediate result size (in million tuples) in DuckDB.
\unfinished{\texttt{TLE}} indicates 900s time limit exceeded. \unfinished{\texttt{OOM}} indicates out-of-memory. 
The best (smallest) value per dataset-query pair is highlighted with a \colorbox{gray!25}{gray background}, and \high{underline} indicates both systems use the same query plan.}
\label{tab:duckdb_max_intermediate_mt}
\vspace{-10pt}
\end{table*}

\subsubsection{Umbra Comparison}
Table~\ref{tab:umbra_cmp} presents the runtime comparison between \sysname{} and three join settings in Umbra. Table~\ref{tab:umbra_max_intermediate_mt} reports the maximum intermediate result size when executing queries in Umbra. \emph{Binary} indicates that Umbra is restricted to binary joins only, while \emph{WCOJ} enforces the use of its built-in worst-case optimal join. \emph{Default} refers to Umbra’s optimizer-chosen plan, which may combine binary joins, WCOJ, or a hybrid of the two. Overall, all queries that can be executed under Umbra’s binary join (35 in total) can also be executed with \sysname{}. Moreover, \sysname{} enables Umbra to finish 10 additional queries that the default setting cannot complete (45 vs. 35). For the queries that can run in both cases, \sysname{} achieves 1.3$\times$ speedups on average (up to 6.1$\times$) and reduces maximum intermediate results by 1.2$\times$ on average (up to 2.1$\times$).

In our analysis, we first focus on comparing \sysname{} with Umbra’s binary join (denoted as UmbraBJ below), as this aligns directly with the primary objective of our work.

Overall, the results are consistent with our observations on DuckDB: \sysname{} outperforms UmbraBJ in similar query cases. When \sysname{} underperforms the default binary join, the performance gap remains marginal and is significantly smaller than that of Umbra’s WCOJ (i.e., Umbra’s WCOJ often times out in cases where the default binary join successfully finishes), underscoring \sysname{}’s robustness and balanced efficiency across diverse workloads. For queries where \sysname{} takes longer time, Orkut Q7 shares the same reason discussed in the DuckDB experiments. The most interesting case is Skitter Q4—Umbra’s query plan, $(R_1\bowtie R_2)\bowtie((R_4\bowtie R_5)\bowtie R_3)$, enables the reuse of intermediate results from $(R_1\bowtie R_2)$ and $(R_4\bowtie R_5)$, whereas \sysname{} currently cannot exploit such reuse. These two queries highlights a potential optimization direction for handling redundant computations during splitting. For the remaining slower queries (Orkut Q4 and Topcats Q11), the cause is similar to DuckDB’s Skitter Q3—the additional split operator increases pipeline length and introduces extra overhead, making one of the subqueries slower than the baseline plan. This suggests room for improvement in the scheduling mechanism when applying \sysname{}.

In summary, similar to the DuckDB results, the UmbraBJ experiments confirm that \sysname{} consistently delivers substantial runtime improvements. We now extend our analysis to examine Umbra’s built-in worst-case optimal join (WCOJ) and its default setting for a more comprehensive comparison.

In general, Umbra’s built-in WCOJ occasionally outperforms \sysname{}, but is often slower than UmbraBJ. This suggests that while WCOJ can exploit specific join patterns effectively (e.g., Q1 triangle and Q6 4-clique), its benefits are not consistent across datasets and queries. For Umbra’s default setting, we observe that the optimizer does not always choose the optimal execution plan—for example, in Skitter Q6 and Topcats Q6, the selected plans result in suboptimal performance. Moreover, in Q7, which joins two triangles $(R_1\bowtie R_2\bowtie R_3)\bowtie(R_4\bowtie R_5\bowtie R_6)$, Umbra’s optimizer applies WCOJ within each triangle and then performs a binary join to combine the intermediate results. This design leverages the strengths of both strategies but also illustrates the inherent complexity and difficulty of hybrid plan selection.

Overall, while Umbra’s built-in WCOJ and its default optimizer can provide advantages in specific scenarios,  \sysname{} can offer consistent and stable improvements on a wide range of queries and datasets. This highlights its effectiveness as a general-purpose strategy for accelerating join workloads.
\begin{table*}[t]
\centering
\small
\begin{tabular}{l l c c c c c c c c c c c}
\toprule
 &  & Q1 & Q2 & Q3 & Q4 & Q5 & Q6 & Q7 & Q8 & Q9 & Q10 & Q11 \\
\midrule
\multirow{4}{*}{WGPB} & \sysname{} & \higher{3.013} & \higher{25.354} & \higher{12.665} & \higher{10.861} & \higher{255} & \higher{10.544} & \higher{24.506} & \unfinished{\texttt{OOM}} & \unfinished{\texttt{TLE}} & \unfinished{\texttt{OOM}} & \higher{39.586} \\
                      & Binary  & 5.43 & 156.777 & \unfinished{\texttt{TLE}} & \unfinished{\texttt{OOM}} & \unfinished{\texttt{OOM}} & \unfinished{\texttt{OOM}} & \unfinished{\texttt{OOM}} & \unfinished{\texttt{OOM}} & \unfinished{\texttt{OOM}} & \unfinished{\texttt{TLE}} & \unfinished{\texttt{TLE}} \\
                      & WCOJ    & 3.774 & \unfinished{\texttt{TLE}} & 17.108 & 12.726 & \wcoj{4.404} & \wcoj{4.426} & 37.652 & \unfinished{\texttt{TLE}} & \unfinished{\texttt{TLE}} & \unfinished{\texttt{TLE}} & 39.716 \\
                      & Default & 5.929 & 158.576 & \unfinished{\texttt{TLE}} & \unfinished{\texttt{OOM}} & \unfinished{\texttt{OOM}} & \unfinished{\texttt{OOM}} & \unfinished{\texttt{OOM}} & \unfinished{\texttt{OOM}} & \unfinished{\texttt{OOM}} & \unfinished{\texttt{TLE}} & \unfinished{\texttt{TLE}} \\
\midrule
\multirow{4}{*}{Orkut} & \sysname{} & \higher{11.767} & \unfinished{\texttt{OOM}} & \unfinished{\texttt{OOM}} & 36.268 & \higher{32.068} & \higher{39.277} & 89.032 & \higher{20.327} & \unfinished{\texttt{OOM}} & \unfinished{\texttt{OOM}} & \unfinished{\texttt{OOM}} \\
                      & Binary  & 13.182 & \unfinished{\texttt{OOM}} & \unfinished{\texttt{OOM}} & \higher{30.719} & 36.294 & \unfinished{\texttt{TLE}} & \higher{58.998} & \unfinished{\texttt{OOM}} & \unfinished{\texttt{OOM}} & \unfinished{\texttt{OOM}} & \unfinished{\texttt{OOM}} \\
                      & WCOJ    & \wcoj{7.833} & \wcoj{722.336} & \wcoj{677.691} & 724.578 & \unfinished{\texttt{TLE}} & 53.473 & \unfinished{\texttt{TLE}} & \unfinished{\texttt{TLE}} & \unfinished{\texttt{TLE}} & \unfinished{\texttt{TLE}} & \unfinished{\texttt{TLE}} \\
                      & Default & 7.824 & \unfinished{\texttt{OOM}} & \unfinished{\texttt{OOM}} & 30.679 & 36.156 & \unfinished{\texttt{TLE}} & 52.707 & \unfinished{\texttt{OOM}} & \unfinished{\texttt{OOM}} & \unfinished{\texttt{OOM}} & \unfinished{\texttt{OOM}} \\
\midrule
\multirow{4}{*}{GPlus} & \sysname{} & \higher{33.458} & \unfinished{\texttt{OOM}} & \unfinished{\texttt{OOM}} & \unfinished{\texttt{TLE}} & \unfinished{\texttt{OOM}} & \unfinished{\texttt{OOM}} & \unfinished{\texttt{TLE}} & \unfinished{\texttt{OOM}} & \unfinished{\texttt{OOM}} & \unfinished{\texttt{OOM}} & \unfinished{\texttt{OOM}} \\
                      & Binary  & 41.520 & \unfinished{\texttt{OOM}} & \unfinished{\texttt{OOM}} & \unfinished{\texttt{OOM}} & \unfinished{\texttt{OOM}} & \unfinished{\texttt{OOM}} & \unfinished{\texttt{TLE}} & \unfinished{\texttt{OOM}} & \unfinished{\texttt{OOM}} & \unfinished{\texttt{OOM}} & \unfinished{\texttt{OOM}} \\
                      & WCOJ    & \wcoj{13.397} & \unfinished{\texttt{TLE}} & \unfinished{\texttt{TLE}} & \unfinished{\texttt{TLE}} & \unfinished{\texttt{TLE}} & \unfinished{\texttt{TLE}} & \unfinished{\texttt{TLE}} & \unfinished{\texttt{TLE}} & \unfinished{\texttt{TLE}} & \unfinished{\texttt{TLE}} & \unfinished{\texttt{TLE}} \\
                      & Default & 13.386 & \unfinished{\texttt{OOM}} & \unfinished{\texttt{OOM}} & \unfinished{\texttt{OOM}} & \unfinished{\texttt{OOM}} & \unfinished{\texttt{OOM}} & \unfinished{\texttt{TLE}} & \unfinished{\texttt{OOM}} & \unfinished{\texttt{OOM}} & \unfinished{\texttt{OOM}} & \unfinished{\texttt{OOM}} \\
\midrule
\multirow{4}{*}{USPatent} & \sysname{} & \high{0.170} & \high{0.472} & \high{0.486} & \high{0.553} & \high{0.339} & \high{1.936} & \high{0.389} & \high{8.097} & \high{3.487} & \high{2.854} & \high{0.764} \\
                      & Binary  & \high{0.172} & \high{0.503} & \high{0.477} & \high{0.489} & \high{0.346} & \high{1.941} & \high{0.378} & \high{7.947} & \high{3.591} & \high{2.818} & \high{0.721} \\
                      & WCOJ    & 0.400 & 1.125 & 2.139 & 1.460 & 2.246 & \wcoj{0.626} & 0.984 & 126.174 & 6.606 & \wcoj{2.423} & 1.351 \\
                      & Default & 0.390 & 0.433 & 0.467 & 0.483 & 0.305 & 1.947 & 0.956 & 8.004 & 3.573 & 2.552 & 0.725 \\
\midrule
\multirow{4}{*}{Skitter} & \sysname{} & \high{0.601} & \high{2.900} & \higher{1.509} & 1.845 & \high{0.671} & \high{310} & \high{2.360} & \high{\unfinished{\texttt{OOM}}} & \high{66.540} & \high{\unfinished{\texttt{OOM}}} & \high{14.068} \\
                      & Binary  & \high{0.425} & \high{2.865} & 1.992 & \higher{0.664} & \high{0.676} & \high{305} & \high{2.226} & \high{\unfinished{\texttt{OOM}}} & \high{70.859} & \high{\unfinished{\texttt{OOM}}} & \high{14.087} \\
                      & WCOJ    & 0.511 & 9.611 & 6.379 & 9.221 & 186.201 & \wcoj{1.365} & 57.504 & \unfinished{\texttt{TLE}} & \unfinished{\texttt{TLE}} & \unfinished{\texttt{TLE}} & 240.686 \\
                      & Default & 0.531 & 2.969 & 1.864 & 0.698 & 0.716 & 306 & 1.358 & \unfinished{\texttt{OOM}} & 68.826 & \unfinished{\texttt{OOM}} & 13.770 \\
\midrule
\multirow{4}{*}{Topcats} & \sysname{} & \higher{2.269} & \higher{16.398} & \higher{12.642} & \higher{5.475} & \higher{5.820} & \higher{8.935} & \higher{10.453} & \higher{62.254} & \higher{311} & \unfinished{\texttt{OOM}} & 366 \\
                      & Binary  & 7.499 & 16.886 & 15.466 & 10.782 & 6.650 & 10.592 & 12.431 & \unfinished{\texttt{OOM}} & \unfinished{\texttt{OOM}} & \unfinished{\texttt{OOM}} & \higher{250} \\
                      & WCOJ    & \wcoj{1.428} & \unfinished{\texttt{TLE}} & 140.463 & 64.064 & 54.212 & \wcoj{4.537} & \unfinished{\texttt{TLE}} & \unfinished{\texttt{TLE}} & \unfinished{\texttt{TLE}} & \unfinished{\texttt{TLE}} & \unfinished{\texttt{TLE}} \\
                      & Default & 1.387 & 17.020 & 15.206 & 10.375 & 6.781 & 10.592 & 5.475 & \unfinished{\texttt{OOM}} & \unfinished{\texttt{OOM}} & \unfinished{\texttt{OOM}} & 249.193 \\
\bottomrule
\end{tabular}
\caption{Runtime (s) in Umbra. \unfinished{\texttt{TLE}} indicates 900s time limit exceeded. \unfinished{\texttt{OOM}} indicates out-of-memory. For comparison, the smaller (better) result between \sysname{} and Binary is highlighted with a \colorbox{gray!25}{gray background}, and \high{underline} indicates \sysname{} and Binary are using the same plan. \wcoj{Bold} indicates that Umbra's built-in WCOJ used less time comparing to \sysname{} and Binary.}
\label{tab:umbra_cmp}
\vspace{-8pt}
\end{table*}

\begin{table*}[t]
\centering
\small
\begin{tabular}{l l c c c c c c c c c c c}
\toprule
 & & Q1 & Q2 & Q3 & Q4 & Q5 & Q6 & Q7 & Q8 & Q9 & Q10 & Q11 \\
\midrule
\multirow{2}{*}{WGPB} 
 & \sysname{} & \higher{559} & \higher{559} & \higher{1026} & \higher{1026} & \higher{2982} & \higher{559} & \higher{559} & \unfinished{\texttt{OOM}} & \unfinished{\texttt{TLE}} & \unfinished{\texttt{OOM}} & \higher{7345} \\
 & Binary & 936 & 936 & \unfinished{\texttt{TLE}} & \unfinished{\texttt{OOM}} & \unfinished{\texttt{OOM}} & \unfinished{\texttt{OOM}} & \unfinished{\texttt{OOM}} & \unfinished{\texttt{OOM}} & \unfinished{\texttt{OOM}} & \unfinished{\texttt{TLE}} & \unfinished{\texttt{TLE}} \\
\midrule
\multirow{2}{*}{Orkut} 
 & \sysname{} & \higher{9450} & \unfinished{\texttt{OOM}} & \unfinished{\texttt{OOM}} & \higher{11803} & \higher{9540} & \higher{16452} & \higher{9450} & \higher{9450} & \unfinished{\texttt{OOM}} & \unfinished{\texttt{OOM}} & \unfinished{\texttt{OOM}} \\
 & Binary & 10913 & \unfinished{\texttt{OOM}} & \unfinished{\texttt{OOM}} & 12489 & 10913 & \unfinished{\texttt{TLE}} & 10913 & \unfinished{\texttt{OOM}} & \unfinished{\texttt{OOM}} & \unfinished{\texttt{OOM}} & \unfinished{\texttt{OOM}} \\
\midrule
\multirow{2}{*}{GPlus} 
 & \sysname{} & \higher{3286} & \unfinished{\texttt{OOM}} & \unfinished{\texttt{OOM}} & \unfinished{\texttt{TLE}} & \unfinished{\texttt{OOM}} & \unfinished{\texttt{OOM}} & \unfinished{\texttt{TLE}} & \unfinished{\texttt{OOM}} & \unfinished{\texttt{OOM}} & \unfinished{\texttt{OOM}} & \unfinished{\texttt{OOM}} \\
 & Binary & 6949 & \unfinished{\texttt{OOM}} & \unfinished{\texttt{OOM}} & \unfinished{\texttt{OOM}} & \unfinished{\texttt{OOM}} & \unfinished{\texttt{OOM}} & \unfinished{\texttt{TLE}} & \unfinished{\texttt{OOM}} & \unfinished{\texttt{OOM}} & \unfinished{\texttt{OOM}} & \unfinished{\texttt{OOM}} \\
\midrule
\multirow{2}{*}{USPatent} 
 & \sysname{} & \high{82} & \high{82} & \high{362} & \high{239} & \high{82} & \high{1539} & \high{82} & \high{1381} & \high{1381} & \high{1381} & \high{362} \\
 & Binary & \high{82} & \high{82} & \high{362} & \high{239} & \high{82} & \high{1539} & \high{82} & \high{1381} & \high{1381} & \high{1381} & \high{362} \\
\midrule
\multirow{2}{*}{Skitter} 
 & \sysname{} & \high{444} & \high{454} & 454 & 496 & \high{444} & \high{237177} & \high{444} & \high{\unfinished{\texttt{OOM}}} & \high{3999} & \high{\unfinished{\texttt{OOM}}} & \high{18786} \\
 & Binary & \high{444} & \high{454} & 454 & \higher{444} & \high{444} & \high{237177} & \high{444} & \high{\unfinished{\texttt{OOM}}} & \high{3999} & \high{\unfinished{\texttt{OOM}}} & \high{18786} \\
\midrule
\multirow{2}{*}{Topcats} 
 & \sysname{} & \higher{1944} & \higher{1944} & \higher{1944} & \higher{1944} & \higher{1944} & \higher{1944} & \higher{1944} & \higher{2371} & \higher{74733} & \unfinished{\texttt{OOM}} & \higher{199772} \\
 & Binary & 2663 & 2663 & 2663 & 2663 & 2663 & 2663 & 2663 & \unfinished{\texttt{OOM}} & \unfinished{\texttt{OOM}} & \unfinished{\texttt{OOM}} & 238985\\
\bottomrule
\end{tabular}
\caption{Max intermediate result size(million tuples) in Umbra.
\unfinished{\texttt{TLE}} indicates 900s time limit exceeded. \unfinished{\texttt{OOM}} indicates out-of-memory. 
The best (smallest) value per dataset-query pair is highlighted with a \colorbox{gray!25}{gray background}, and \high{underline} indicates using the same plan.}
\label{tab:umbra_max_intermediate_mt}
\vspace{-10pt}
\end{table*}

\subsection{Effectiveness Study}\label{sec:effective}

In this section, we use DuckDB to evaluate the effectiveness of each strategy proposed in \sysname{}.

\begin{table*}[t]
\centering
\small
\begin{tabular}{l ccc ccc ccc ccc}
\toprule
\multirow{2}{*}{Method} 
& \multicolumn{3}{c}{WGPB} 
& \multicolumn{3}{c}{Orkut} 
& \multicolumn{3}{c}{GPlus} 
& \multicolumn{3}{c}{Topcats} \\
\cmidrule(lr){2-4} \cmidrule(lr){5-7} \cmidrule(lr){8-10} \cmidrule(lr){11-13}
& Q1 & Q2 & Q5 & Q1 & Q2 & Q5 & Q1 & Q2 & Q5 & Q1 & Q2 & Q5 \\
\midrule
DuckDB Default 
& \unfinished{\texttt{TLE}} & 59.230 & \unfinished{\texttt{TLE}} 
& 19.200 & \unfinished{\texttt{OOM}} & \unfinished{\texttt{OOM}} 
& 34.770 & \unfinished{\texttt{OOM}} & \unfinished{\texttt{OOM}} 
& 7.800 & 30.660 & \unfinished{\texttt{TLE}} \\
\midrule
Single Split
& 8.262 & 74.667 & 107.836 
& 17.756 & \unfinished{\texttt{OOM}} & 181.599 
& 46.750 & \unfinished{\texttt{OOM}} & \unfinished{\texttt{OOM}} 
& 6.795 & 70.766 & 36.136 \\
\midrule
Co-split
& 22.276 & 59.945 & 61.923 
& 34.997 & \unfinished{\texttt{OOM}} & 74.049 
& 34.724 & \unfinished{\texttt{OOM}} & \unfinished{\texttt{OOM}}
& 10.577 & 26.947 & 355 \\
\midrule
Co-split + Split-Set-Selection Strategies
& 3.824 & 29.879 & 59.876
& 15.076 & \unfinished{\texttt{OOM}} & 56.548 
& 27.156 & \unfinished{\texttt{OOM}} & \unfinished{\texttt{TLE}} 
& 4.525 & 24.056 & 9.207 \\
\bottomrule
\end{tabular}
\caption{Runtime(s) effects on Q1, Q2, and Q5 across datasets. 
\unfinished{\texttt{TLE}} indicates 900s time limit exceeded. 
\unfinished{\texttt{OOM}} indicates out-of-memory.}
\label{tab:effectiveness_combined}
\vspace{-8pt}
\end{table*}

\subsubsection{Split Threshold}

To evaluate the effectiveness of the split threshold proposed in Section~\ref{sec:split_threshold}, we conduct experiments on query Q1 (triangle query) over the GPlus and Topcats datasets, for the remaining datasets, the results are similar. Specifically, we fix the splitting attribute on a single table and vary the threshold 
to compare their impact on join execution time. Note that threshold=0 corresponds to the binary join plan.

In Figure~\ref{fig:threshold_vs_time_tuple}, as the threshold increases from 0, the execution time and the maximum intermediates  first decrease and then increase again. The selected thresholds in \sysname{} (highlighted with slashes) consistently achieve near-optimal performance across datasets. 
\begin{figure}[t]
  \centering
  \includegraphics[width=\linewidth]{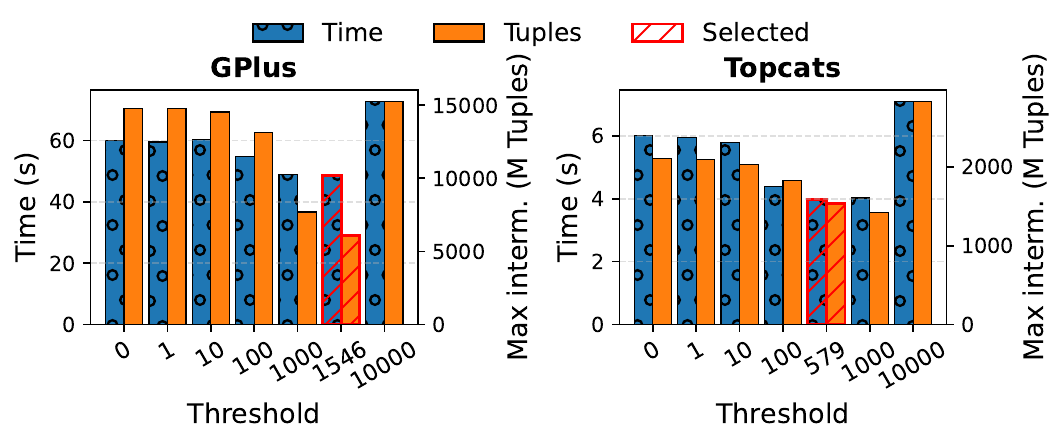}
  \captionsetup{font=small,skip=3pt}
  \caption{Threshold vs Time (s) in Q1}
  \label{fig:threshold_vs_time_tuple}
  
\vspace{-3pt}
\end{figure}

\subsubsection{Split Schedule}

\begin{figure}[t]
  \centering
  \includegraphics[width=\linewidth]{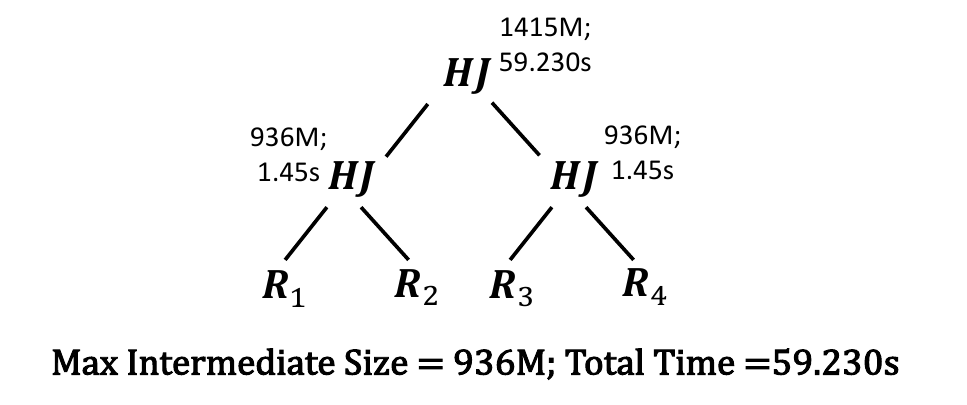}
  \captionsetup{font=small,skip=3pt}
  \caption{Per-join performance breakdown of the most efficient plan on Topcats Q2. Labels in the form “x;y” indicate that a subtree outputs x million tuples and runs for y seconds.}
  \label{fig:case_study_Q2}
\end{figure}
\begin{figure}[t]
  \centering
  \includegraphics[width=\linewidth]{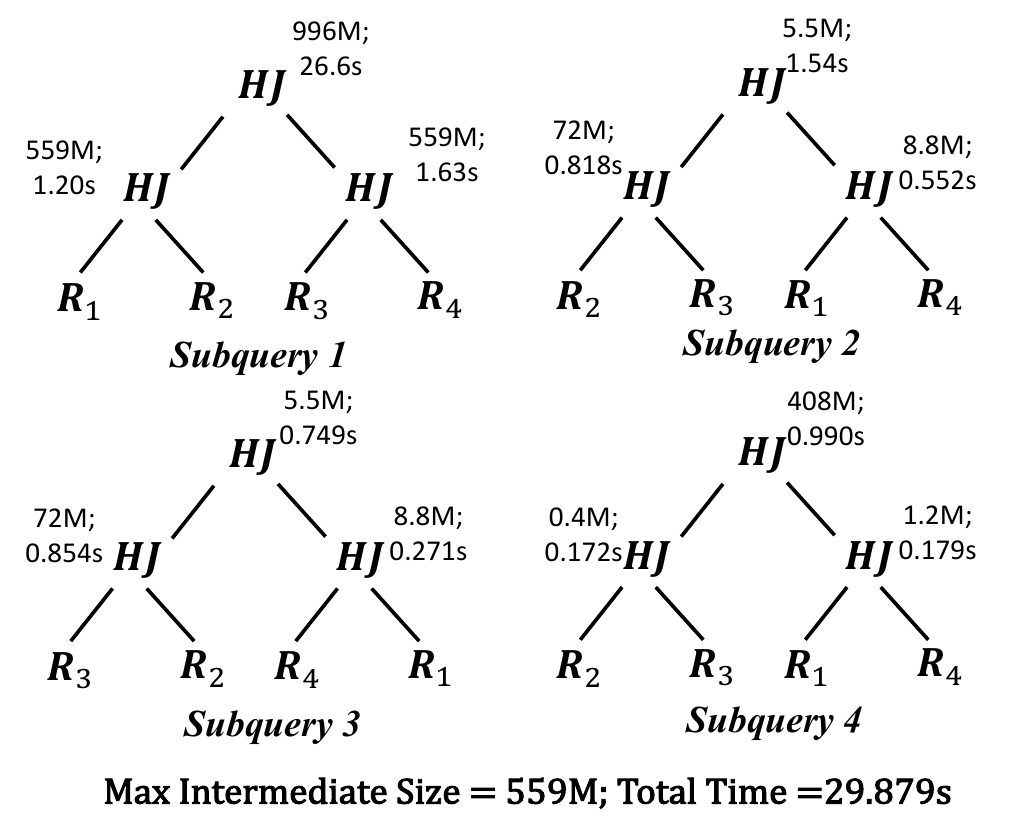}
  \captionsetup{font=small,skip=3pt}
  \caption{Per-join performance breakdown of SplitJoin on Topcats Q2. Labels in the form “x;y” indicate that a subtree outputs x million tuples and runs for y seconds.}
  \label{fig:case_study_Q2_split}
\end{figure}
To evaluate the effectiveness of the proposed strategies, namely \emph{co-split}, and \emph{split-set-selection} strategies, we conduct experiments on three representative 
queries: Q1 (triangle), Q2 (rectangle), and Q5 (diamond). The experiments are performed on the WGPB, Orkut, 
GPlus, and Topcats datasets. For the USPatent and Skitter datasets, splitting is not applied to these queries 
and thus they are excluded from this evaluation. We consider three configurations in total: 
(i) \textbf{config1}: only split on a single table each time, and in this case, we will split on the tables and attributes that config3 chooses, 
(ii) \textbf{config2}: \emph{co-split} on a fixed split set, and 
(iii) \textbf{config3}: config2 enhanced with \emph{Split-Set-Selection} strategies mentioned in Section~\ref{sec:split_schedule}.

As shown in Table~\ref{tab:effectiveness_combined}, our final configuration, config3, achieves better performance compared with both config1 and config2.
When compared with config1, the improvement mainly comes from a reduction in the number of generated subqueries—config1 produces 4, 16, and 16 subqueries for Q1, Q2, and Q5, respectively, while config3 only produces 2, 4, and 4 subqueries.
When compared with config2, the gain primarily stems from a smaller upper bound on intermediate results with appropriate co-split set selection, which helps reduce execution overhead. For example, in Q1, splitting on a lighter join rather than a larger one significantly reduces the size of intermediate results.

\subsection{Case Study}\label{sec:case}
To provide a deeper understanding of how our approach works, we present a case study on query Q2 in the WGPB dataset. 
Q2 in Figure~\ref{fig:queries} can be expressed as $R_1(X,Y)\bowtie R_2(Y,W)\bowtie R_4(X,Z)\bowtie R_3(Z,W)$ where the original query is split into four subqueries by splitting on $R_1\bowtie R_2$ and $R_3\bowtie R_4$.

We compare the baseline plan in DuckDB with the plan in DuckDB + \sysname{}, as it reveals the advantage of split-based execution: \sysname{} avoids generating large intermediates by isolating heavy keys, leading to faster overall runtime (i.e., overall 29.9s with \sysname{} vs. 59.2s without \sysname{}).

Figure~\ref{fig:case_study_Q2} shows that the baseline plan produces a massive intermediate result when joining $R_3$ with $R_4$ and $R_1$ with $R_2$ (both of them expanding to more than 936 million tuples). This blowup significantly increases the number of hash probes in the hash table built by $(R_3\bowtie R_4)$, and it also increases the time consumption on building hash table on $(R_3\bowtie R_4)$, ultimately becoming the bottleneck of the entire join plan, as confirmed by the time breakdown of each step.

In contrast, \sysname{} decomposes Q2 into four subqueries (Table~\ref{fig:case_study_Q2_split}), each with more balanced join sizes. The largest intermediate relation is reduced to about 559 million tuples, which makes the final join steps substantially more efficient and effectively alleviates the bottleneck observed in the baseline. As a result, \sysname{} achieves a more stable execution profile and improves overall performance for Q2. This case study highlights the performance bottlenecks of the baseline and demonstrates how \sysname{} avoids skewed joins and achieves better efficiency.

\vspace{-5pt}

\section{Conclusion}\label{sec:conclusion}

In this work, we presented \sysname{}, a lightweight framework that introduces the split operator to minimize intermediate join results and mitigate data skew. Implemented as a non-intrusive front-end layer, \sysname{} integrates seamlessly with existing systems and delivers substantial performance gains on real-world skewed workloads. This work takes a first step toward making split-based strategies practical, opening new opportunities for adaptive cost models and broader integration into modern query engines.

\bibliographystyle{ACM-Reference-Format}
\bibliography{ref}

\appendix

\newpage
\section{Analysis of Algorithm \ref{algorithm:BinaryJoin}}
\label{sec:TheoryAppendix}

\newcommand{\dl}{\text{ :- }}
\newcommand{\AGM}[1]{$\mathsf{AGM}(#1)$}
\newcommand{\mAGM}[1]{\mathsf{AGM}(#1)}

\subsection{Correctness}

Throughout this section, we will assume that the query $Q$ is connected, meaning there are no cartesian products.
This limitation can easily be remedied by computing each connected sub-query individually and then taking the cartesian product.

\begin{theorem}
    Algorithm \ref{algorithm:BinaryJoin} terminates and returns the correct output.
\end{theorem}
\begin{proof}
    
    We will show the following three statements.
    \begin{enumerate}
        \item The algorithm terminates
        \item When the algorithm terminates, $\mathcal{C}$ has one element
        \item The unique element in $\mathcal{C}$ is the correct output of the query
    \end{enumerate}
    
    \textbf{1)} The loop on line 2 terminates because the number of iterations is bounded by the number of attributes. The loop on line 5 terminates because the number of possible light joins is bounded by the number of edges. The loop on line 8 terminates since the number of components in $\mathcal{C}$ is finite. \\
    \textbf{2)} Suppose there are more than one component in $\mathcal{C}$. One of these must have been discovered in the last iteration of the loop on line 2. Since the components were not merged, they cannot have any overlap. Since $Q$ is connected, this means there is at least one attribute that is not in any component, because otherwise some components would have intersected. Some such attribute would have either caused another iteration of the loop on line 2, or would have been discovered by a light join from some component, so its existence is a contradiction. \\
    \textbf{3)} We will show that every relation is joined into the output, which makes the output be correct. Let $e$ be an arbitrary relation. Let $x$ be the attribute on the light end of $e$. $x$ will be integrated into some component at some point, by the same argument as in 2). This component would contain the relation $e$, since the algorithm exhausts light joins. Since all components are merged, the output will be a join that includes the relation $e$.
\end{proof}

\subsection{Bounding the Size of Intermediate Results}

We will now show that Algorithm \ref{algorithm:BinaryJoin} matches the AGM bound. We remark that this works for join queries with binary relations, where all cardinalities are at most $N$. We will think of a query $Q$ as its query graph $(V,E)$, which we described previously.

\begin{definition}[Fractional Vertex Packing]
    Let $Q=(V,E)$ be a natural join query. A fractional vertex packing $\mathbf{u}$ is an assignment to every vertex $x\in V$ in the graph a weight $u_x\in[0,1]$ such that for every edge $\{x,y\}$, we have $u_x+u_y \leq 1$.
\end{definition}

We denote by $VP$ be the set of all possible vertex packings of $Q$.

\begin{theorem}[AGM Bound]
\label{theorem:AGM}
    Let $Q$ be a join query with binary relations. Let $I$ be a database instance such that every relation has size at most $N$. Then
    \[
    |Q(I)| \leq \mathcal{AGM}(Q) := \max_{\mathbf{u}\in VP} N^{\sum_{x} u_x}
    \]
\end{theorem}

The maximum fractional vertex packing is the same as the minimum fractional edge cover, by duality of linear programs. Therefore, the the maximum $\mathbf{v}\in VP$ has a total weight of $\rho$.

The goal of a worst-case optimal algorithm is to match this bound, i.e. make sure that no intermediate result has size more than $\mathcal{AGM}(Q)$.

\begin{definition}
    Given a query $Q=(V,E)$, and a subset of $I\subseteq V$, define the subquery $Q_I=(I,E_I)$, where $E_I=\{ \pi_IR | R\in E \text{ where } \exists x\in R \text{  s.t. } x\in I\}$. We will call subqueries on this form \textbf{induced subgraph queries}.
\end{definition}

The next lemma means that any intermediate relation that has the form of an induced subgraph query can never be too big. This is a key idea in the area of worst case optimal joins.

\begin{lemma}\label{lemma:InducedSubgraph}
    Let $Q=(V,E)$ be a query. Then, for any $I\subseteq V$
    \[ \mAGM{Q_I} \leq \mAGM{Q} \]
\end{lemma}
\begin{proof}
    Consider any fractional vertex packing $\mathbf{u}_{Q_I}$ in $Q_I$. We can create a feasible fractional vertex packing $\mathbf{u}_Q$ for $Q$ by taking $\mathbf{u}_{Q_I}$, and extending it to $Q$ by assigning a weight $0$ to any vertex in $V-I$. These fractional vertex packings have the same weight. This means that the maximal fractional vertex packing in $Q$ is at least as big as the maximal fractional vertex packing in $Q_I$.
\end{proof}

\begin{lemma}\label{lemma:LightExpansionAGM}
    Let $Q=(V,E)$ be a query where all cardinalities are at most $N$. Let $I$ be an intermediate result from $Q$, that contains $k_I$ unique attributes, that is obtained by starting from any relation and doing light joins. Then, $|I| \leq N^{k_I/2} \leq \mathcal{AGM}(Q)$.
\end{lemma}

\begin{proof}
    The first inequality follows from the fact that a light join that reaches a new attribute grows the intermediate size by at most a factor $\sqrt{N}$, and that we start with a single relation, whose size is at most $N^{2\times 1/2}$. For the second inequality, let $k$ be the number of vertices in $Q$. Clearly $N^{k_I/2}\leq N^{k/2}$. The inequality $N^{k/2}\leq \mathcal{AGM}(Q)$ then follows from the fact that we can obtain a feasible vertex packing by assigning weight $1/2$ to every vertex.
\end{proof}

\begin{lemma}\label{lemma:MergeStepFine}
    Let $Q=(V,E)$ be a join query and let $I$ be a database instance. Let $U\subset V$. Suppose that
    \begin{enumerate}
        \item for every $Y\in V-U$ there exists a relation $R(X,Y)\in E$ for some $X\in U$
        \item any relation $R(X,Y)\in E$ where $X\in U,Y\in V-U$ is light in $X$
    \end{enumerate}
    Let $T$ be an intermediate relation obtained by computing $Q_{U}$ and then joining with all relations $R(X,Y)\in E$ where $X\in U$, $Y\in V-U$. Then
    \[ |T| \leq \mathcal{AGM}(Q) \]
\end{lemma}
\begin{proof}
    We will show that there exists an instance $I'$ such that $|T| \leq Q(I')$. This proves the lemma by Theorem \ref{theorem:AGM}. For some $R\in E$, denote by $R_I$ and $R_{I'}$ the relation-instances in $I$ and $I'$, respectively.

    For any relation $R(X,Y)$ where $X,Y\in U$, let $R_{I'}=R_I$. For any $R(X,Y)$ where $X,Y\in V-U$, $R_{I'}$ is a product-relation $[\sqrt{N}]\times[\sqrt{N}]$. Finally, consider relations $R(X,Y)$ where $X\in U$ and $Y\in V-U$. For each $x\in \pi_X R_I$, add to $R_{I'}$ the tuples $(x,1),\dots,(x,\deg_R(x))$.

    The cardinality constraints are satisfied for $I'$, since for each $R$ with at least $1$ variable in $U$, each tuple in $R_I$ adds one tuple in $R_{I'}$.
    
    We now show $|T| \leq Q(I')$. Let $\Psi_I:\pi_UT\rightarrow\mathbb{Z}_{\geq 0}$, where $\Psi_I(t)=|\{ t'\in T, \pi_{U}t'=t \}|$ (resp. let $\Psi_{I'}:\pi_UQ(I')\rightarrow\mathbb{Z}_{\geq 0}$, where $\Psi_{I'}(t)=|\{ t'\in Q(I'), \pi_{U}t'=t \}|$). It is clear that the output sizes now can be written as the sum $\sum_{t\in \pi_{U}T}\Psi_I(t)$ (resp. $\sum_{t\in \pi_{U}Q(I')}\Psi_{I'}(t)$). It remains to show that for all $t$, $\Psi_I(t) \leq \Psi_{I'}(t)$.
    
    We claim the following
    \[
    \Psi_{I'}(t) = \prod_{y\in V-U}\min_{R(x,y):x\in U}\deg_R(x)
    \]
    This is because of how the relations are picked, all the relations going to some $y\in V-U$ cover the same values. Additionally, from the assumption $(2)$, $\deg_R(x)\leq\sqrt{N}$, which means that the relations with both variables in $V-U$ do not do any filtering. Finally, by the assumption $(1)$, there exists a light join reaching $y$, meaning the minimization above is well defined.
    
    On the other hand, $\Psi_{I}(t)$ is upper bounded by the above. This is because these minimum degree relations are needed to cover any output tuple using $t$. The relations of larger degrees may then act as filters, reducing $\Psi_I(t)$.
\end{proof}

We will finally prove the main theorem. 

\TheoryAlgMatchAGM
\begin{proof}
    This instantiation creates $2^l$ subproblems, where $l$ is the number of relations. For each subproblem, we have a join plan, where every relation appears once. Hence, the number of binary joins performed only depends on the query. Using that a binary join can be computed in time $O(IN+OUT)$, it only remains to argue that all intermediate relations have size at most $\mathcal{AGM}(Q)$.
    
    By Lemma \ref{lemma:LightExpansionAGM}, the intermediate result obtained from light joins, i.e. between lines 3 and 5 in the algorithm, is bounded by AGM. Additionally, when the merging step is finished, at line 9, the intermediate result has the form of an induced subquery, which is bounded by AGM by Lemma \ref{lemma:InducedSubgraph}. It therefore remains to argue that the the size is bounded by $\mathcal{AGM}(Q)$ also for intermediates obtained during any iteration of the loop on line 6. 
    
    Let $V'$ be the set of attributes occurring in such an intermediate $T$, and let $U$ be the set of attributes that were first included in any intermediate by $T$. Any attributes in $V'-U$ must have occurred in another intermediate relation that was created before $T$. $T$ includes all relations in $Q_{U}$ and additionally, any relation $R(X,Y)$ where $X\in U$, $Y\in V'$, where $R$ is light in $X$. Furthermore, there cannot exist any $R'(X',Y')$ where $X'\in V'-U$, $Y'\in U$ that is light in $X'$, because then an earlier intermediate result would have had a light join to $Y'$.
    
    Therefore, any such intermediate $T$ has the form described in Lemma \ref{lemma:MergeStepFine}, in the query $Q_{V'}$. The application of this Lemma finishes this proof. 
\end{proof}

\end{document}